\documentclass[12pt,twoside]{article}

 \usepackage{float}
\usepackage{graphicx}
\usepackage{epstopdf}
\usepackage{graphicx}
\usepackage{epic}
\usepackage{multirow}
\usepackage{tikz}
\usepackage{xcolor}
\usepackage{makecell}
\usetikzlibrary{arrows,shapes,chains}

\renewcommand{\paragraph}{\roman{paragraph}}
\usepackage[a4paper]{geometry}
\makeatletter
\renewcommand\title[1]{\gdef\@title{\reset@font\Large\bfseries #1}}
\renewcommand\section{\@startsection {section}{1}{\z@}%
                                   {-3.5ex \@plus -1ex \@minus -.2ex}%
                                   {2.3ex \@plus.2ex}%
                                   {\normalfont\large\bfseries}}
\renewcommand\subsection{\@startsection{subsection}{2}{\z@}%
                                     {-3ex\@plus -1ex \@minus -.2ex}%
                                     {1.5ex \@plus .2ex}%
                                     {\normalfont\normalsize\bfseries}}
\renewcommand\subsubsection{\@startsection{subsubsection}{3}{\z@}%
                                     {-2.5ex\@plus -1ex \@minus -.2ex}%
                                     {1.5ex \@plus .2ex}%
                                     {\normalfont\normalsize\bfseries}}

\def\@runningauthor{}\newcommand{\runningauthor}[1]{\def\runningauthor{#1}}
\def\@runningtitle{}\newcommand{\runningtitle}[1]{\def\runningtitle{#1}}

\renewcommand{\ps@plain}{%
\renewcommand{\@evenhead}{\footnotesize\scshape \hfill\runningauthor\hfill}
\renewcommand{\@oddhead}{\footnotesize\scshape \hfill\runningtitle\hfill}}

\newcommand{\F}{\mathbb{F}}

\newcommand\undermat[2]{%
  \makebox[0pt][l]{$\smash{\underbrace{\phantom{%
    \begin{matrix}#2\end{matrix}}}_{\text{$#1$}}}$}#2}

\g@addto@macro\bfseries{\boldmath}

\makeatother

\usepackage{amsthm,amsmath,amssymb}
\usepackage{cite}
\usepackage{graphicx}

\usepackage[colorlinks=true,citecolor=black,linkcolor=black,urlcolor=blue]{hyperref}

\theoremstyle{plain}
\newtheorem{theorem}{Theorem}[section]

\newtheorem{lem}[theorem]{Lemma}
\newtheorem{cor}[theorem]{Corollary}
\newtheorem{prop}[theorem]{Proposition}

\theoremstyle{definition}

\newtheorem{example}[theorem]{Example}

\theoremstyle{remark}
\newtheorem{remark}[theorem]{Remark}

\runningauthor{}

\date{}

\begin{document}

\title{Characterization and construction of optimal binary linear codes with one-dimensional hull\thanks{The research of Shitao Li and Minjia Shi is supported by the National Natural Science Foundation of China (12071001). The research of Jon-Lark Kim is supported by the National Research Foundation of Korea (NRF) Grant funded by the Korea government (NRF-2019R1A2C1088676).}}
\author{Shitao Li\thanks{Shitao Li is with the School of Mathematical Sciences, Anhui University, Hefei, 230601, China, email: lishitao0216@163.com.}, Minjia Shi\thanks{Minjia Shi is with the Key Laboratory of Intelligent Computing and Signal Processing, Ministry of Education, and also the School of Mathematical Sciences, Anhui University, Hefei, 230601, China; State Key Laboratory of Information Security, Institute of Information Engineering, Chinese Academy of Sciences, Beijing, 100093, China, email: smjwcl.good@163.com}, Jon-Lark Kim\thanks{Jon-Lark Kim is with the Department of Mathematics, Sogang University, Seoul, South Korea, email: jlkim@sogang.ac.kr}}

    \maketitle

\begin{abstract}
The hull of a linear code over finite fields is the intersection of the code and its dual, and linear codes with small hulls have applications in computational complexity and information protection. Linear codes with the smallest hull are LCD codes, which have been widely studied. Recently, several papers were devoted to related LCD codes over finite fields with size greater than 3 to linear codes with one-dimensional or higher dimensional hull. Therefore, an interesting and non-trivial problem is to study binary linear codes with one-dimensional hull with connection to binary LCD codes. The objective of this paper is to study some properties of binary linear codes with one-dimensional hull, and establish their relation with binary LCD codes. Some interesting inequalities are thus obtained. Using such a characterization, we study the largest minimum distance $d_{one}(n,k)$ among all binary linear $[n,k]$ codes with one-dimensional hull. We determine the largest minimum distances $d_{one}(n,n-k)$ for $ k\leq 5$ and $d_{one}(n,k)$ for $k\leq 4$ or $14\leq n\leq 24$. We partially determine the exact value of $d_{one}(n,k)$ for $k=5$ or $25\leq n\leq 30$.
\end{abstract}
{\bf Keywords:} Hull, binary LCD code, minimum distance, building-up construction.\\

\noindent{\bf Mathematics Subject Classification} 94B05 15B05 12E10

\section{Introduction}

The hull of a linear $[n,k]$ code $C$ is defined as ${\rm Hull}(C)=C\cap C^{\perp},$ which was introduced in 1990 by Assmus and Key \cite{A-hull-DM} to classify finite projective planes.
Suppose that the dimension of ${\rm Hull}(C)$ is $\ell$. If $\ell=0$, that is, $C\cap C^{\perp}=\{\textbf 0\}$, then the code $C$ is a linear complementary dual (LCD) code.
If $\ell=k$, that is, $C\subseteq C^{\perp}$, then the code $C$ is a self-orthogonal code. It has been shown that the hull determines the complexity of the algorithms for checking permutation equivalence of two linear codes and for computing the automorphism group of a linear code \cite{J-p-group,Sen-p-e-codes,Sen-S-auto-group}. It turns out that most of the algorithms do not work if the hull is large. Therefore, studying linear codes with small hulls is helpful for these computations. Further, they also have been employed to construct entanglement-assisted quantum error-correction codes \cite{Dcc-EAqecc,zhu-2}.

The smallest dimension of the hull of a linear code is 0, i.e., an LCD code, which was introduced by Massey \cite{LCD-Massey} in order to provide an optimum linear coding
solution for the two-user binary adder channel. Sendrier \cite{LCD-is-good} showed that LCD codes meet the asymptotic Gilbert-Varshamov bound. In 2016, Carlet and Guilley \cite{lcd-appl} investigated an application of binary LCD codes against Side-Channel Attacks (SCA) and Fault Injection Attack (FIA). The study of LCD codes has thus become a hot topic and the reader is referred to \cite{binaryLCD,Shi-1,234-lcD} for recent papers. An interesting result is that Carlet {\em et al.} \cite{LCD-equivalent} showed that any code over $\F_q$ is equivalent to some Euclidean LCD code for $q>3$. This motivates us to study LCD codes, especially binary LCD codes.
Let $d_{ LCD}(n,k)$ denote the largest minimum distance among all binary LCD $[n,k]$ codes. Araya, Harada, and Saito {\em et al.} have made a lot of contributions on the characterization and classification of binary LCD codes. Specifically, the exact value of $d_{ LCD}(n, k)$ for $n\leq 24$ was determined in \cite{bound-12,HS-BLCD-1-16,AH-BLCD-17-24}. The exact value of $d_{LCD}(n,k)$ for $25 \leq n\leq 40$ was partially determined in \cite{2-LCD-30,B-LCD-40,2-LCD-40,LCD-code-Li}. The exact values of $d_{LCD}(n,k)$ and $d_{LCD}(n,n-k)$ for $k\leq 5$ was studied in \cite{LCD-lp,bound-12,HS-BLCD-1-16,AH-BLCD-17-24,AHS-BLCD,ter-11-19}.

The second smallest dimension of the hull of a linear code is 1, i.e., a linear code with one-dimensional hull. Let $d_{one}(n,k)$ denote the largest minimum distance among all binary linear $[n,k]$ codes with one-dimensional hull. Li and Zeng \cite{Li-one-hull} constructed some binary linear codes with one-dimensional hull for $n=8,9,10$ by employing quadratic number fields, partial difference sets, and difference sets. Kim \cite{Kim_preprint} determined the exact value of $d_{one}(n,k)$ for $1 \le k \le n\leq 13$ by a building-up construction. Mankean and Jitman \cite{binary-hull-2} determined the exact value of $d_{one}(n,2)$.
For more related work, readers can refer to \cite{Li-DM,one-qian-ccds,one-qian,Li-DCC}.

 Carlet {\em et al.} \cite{LCD-equivalent} proved that any linear code over $\mathbb F_q$ $(q > 3)$ is equivalent to an Euclidean LCD code.
Consequently, a linear $[n, k, d]$ Euclidean LCD code over $\mathbb F_q$
with $q > 3$ exists if there is a linear $[n, k, d]$ code over $\mathbb F_q$. Recently, Chen \cite{ChenIT} proved that an LCD code over $\F_{2^s}$ ($s\geq 2$) is equivalent to a linear code with one-dimension hull under a weak condition. Kim \cite{Kim_preprint} proved that binary LCD $[n, k]$ codes can produce binary linear $[n+2, k+1]$ codes with one-dimensional hull and did not prove that the converse can be true. On the other hand, our study on binary linear codes with one-dimensional hull is also worth of studying because such codes sometimes have better minimum distances than binary LCD codes with the same length and dimension. Combining all these facts, we ask the natural question of how binary LCD codes are related to binary linear codes with one-dimensional hull. We solve this problem in this paper.

In this paper, we study some properties of binary linear codes with one-dimensional hull, and establish the connection between such codes and binary LCD codes. Some interesting inequalities for $d_{LCD}(n,k)$ and $d_{one}(n,k)$ are obtained. Using the building-up construction in \cite{Kim_preprint} and these inequalities, we extend Kim's results to lengths up to $30$.
Further, we determine the exact values of $d_{one}(n,k)$ and $d_{one}(n,n-k)$ for $k\leq 5$ except for some special types.

This paper is organized as follows. In Section 2, we give some preliminaries. In Section 3, we establish the connection between binary LCD codes and binary linear codes with one-dimensional hull. In Section 4, we study some properties of binary linear codes with one-dimensional hull. In Section 5, we introduction a building-up construction that helps us determined the values of $d_{one}(n,k)$ for $n\leq 30$. In Section 6, we characterize the values of $d_{one}(n,k)$ and $d_{one}(n,n-k)$ for $k \leq 5$. In Section 7, we conclude the paper.

\section{Preliminaries}
\subsection{Binary linear codes and some bounds}
Let $\F_2$ denote the finite field with 2 elements.
For any ${\bf x}\in \F_2^n$, the {\em support} of ${\bf x}=(x_1,x_2,\ldots,x_n)$ is defined as $supp({\bf x})=\{i~|~x_i=1, 1\leq i\leq n\}.$
The {\em (Hamming) weight} ${\rm wt}({\bf x})$ of {\bf x} is the number of nonzero coordinates of {\bf x}, so ${\rm wt}({\bf x})=|supp({\bf x})|$. The {\em distance} between two vectors ${\bf x}$ and ${\bf y}$ is $d({\bf x}, {\bf y})={\rm wt}({\bf x}- {\bf y})$.
The {\em minimum distance} of $C$ is defined by $\min \{ d({\bf x}, {\bf y})~|~  {\bf x},{\bf y} \in C~{\rm and}~ {\bf x}\neq {\bf y}\}$. A binary {\em linear $[n,k]$ code} is a $k$-dimensional subspace of $\F_2^n$. A vector of a binary linear $[n,k]$ code is called a {\em codeword}.
A binary {\em linear $[n,k,d]$ code} $C$ is a binary linear $[n,k]$ code with minimum distance $d$. A {\em generator matrix} for a binary linear $[n, k]$ code $C$ is any
$k\times n$ matrix $G$ whose rows form a basis for $C$.
For any set of $k$ independent columns of a generator matrix $G$, the corresponding
set of coordinates forms an {\em information set} for $C$.

The dual code $C^{\perp}$ of a binary linear $[n,k]$ code $C$ is defined as
$$C^{\perp}=\{\textbf y\in \F_2^n~|~\textbf x\cdot \textbf y=0, {\rm for\ all}\ \textbf x\in C \},$$
where $\textbf x\cdot \textbf y=\sum_{i=1}^n x_iy_i$ for $\textbf x = (x_1,x_2, \ldots, x_n)$ and $\textbf y = (y_1,y_2, \ldots, y_n)\in \F_2^n$. A {\em parity-check matrix} for a linear code $C$ is a generator matrix for the dual code $C^\perp$.
The {\em hull} of a binary linear code $C$ is defined as
${\rm Hull}(C)= C \cap C^\perp.$

It is well-known that the {\em Griesmer bound} \cite[Chap. 2, Section 7]{Huffman} on a binary linear $[n,k,d]$ code is given by $n\geq \sum_{i=0}^{k-1}\left\lceil \frac{d}{2^i}\right\rceil,$ where $\lceil a\rceil$ is the least integer greater than or equal to $a$. A binary $[n,k,d]$ code $C$ is said to a {\em Griesmer code} if $C$ meets the Griesmer bound, i.e., $n=\sum_{i=0}^{k-1}\left\lceil \frac{d}{2^i}\right\rceil.$
The {\em sphere-packing bound} on a binary linear $[n,k,d]$ code is given by
$$2^k\leq \frac{2^n}{\sum_{i=0}^{\lfloor \frac{d-1}{2}\rfloor}\left(\begin{array}{c}
                                                                 n \\
                                                                 i
                                                               \end{array}
\right)},$$
where $\lceil a\rceil$ is the greatest integer less than or equal to $a$.
A vector ${\bf x} = (x_1, x_2, \ldots , x_n)\in \F^n_2$ is {\em even-like} if $\sum_{i=1}^nx_i=0$ and is {\em odd-like} otherwise. A binary code is said to be {\em even-like} if it has only even-like codewords, and is said to be {\em odd-like} if it is not even-like \cite{Huffman}. A binary vector is even-like if and only if it has even weight; so the concept of even-like vectors is indeed a generalization of even weight binary vectors \cite[p. 12]{Huffman}.

\subsection{Characterization of binary linear codes with small hulls}
Carlet {\em et al.} \cite{binaryLCD} presented a new characterization of binary LCD codes, and solved a conjecture proposed by Galvez {\em et al.} \cite{bound-12} on the minimum distance of binary LCD codes. We introduce the new characterization of binary LCD codes as follows.

\begin{theorem}{\rm\cite[Theorem 3]{binaryLCD}}\label{odd-like-lCD}
Let $C$ be an odd-like binary linear $[n,k]$ code. Then $C$ is LCD if and only if there exists a basis $ {\bf c}_1, {\bf c}_2,\ldots, {\bf c}_k$ of $C$ such that for any $i, j\in \{1, 2, \ldots , k\}$, ${\bf c}_i \cdot {\bf c}_j$ equals $1$ if $i = j$ and equals $0$ if $i \neq j$.
\end{theorem}

\begin{theorem}{\rm\cite[Lemma 7]{binaryLCD}}\label{even-like-lCD}
Let $C$ be an even-like binary linear $[n,k]$ code. Then $C$ is LCD if and only if $k$ is even and there exists a basis ${\bf c}_1, {\bf c}'_{1},\ldots, {\bf c}_{\frac{k}{2}}, {\bf c}'_{\frac{k}{2}}$ of $C$ such that for
any $i, j\in \{1, 2, \ldots , \frac{k}{2}\}$, the following conditions hold\\
$(i)$ $ {\bf c}_i\cdot  {\bf c}_i= {\bf c}'_i\cdot  {\bf c}'_i=0$;\\
$(ii)$ $ {\bf c}_i\cdot  {\bf c}'_j=0$, for $i\neq j$;\\
$(iii)$ $ {\bf c}_i\cdot  {\bf c}'_i=1$.\\
$(iv)$ $c_{i,1}=c'_{i,1}$, where ${\bf c}_i=(c_{i,1},\ldots,c_{i,n})$ and ${\bf c}'_i=(c'_{i,1},\ldots,c'_{i,n})$.
\end{theorem}

\begin{lem}{\rm \cite[Proposition 1]{Li-one-hull}}\label{lem-2.3}
Let $C$ be a binary linear $[n,k]$ code with a generator matrix $G$. Then $C$ has $\ell$-dimensional hull if and only if $\ell=k-{\rm rank}(GG^T).$
\end{lem}

\begin{theorem}\label{thm-one}
Let $C$ be a binary linear $[n,k]$ code. Then $C$ is an odd-like (resp. even-like) binary linear code with one-dimensional hull if and only if there exists a basis $ {\bf c}_1, {\bf c}_2,\ldots, {\bf c}_k$ of $C$ such that the code generated by ${\bf c}_1,{\bf c}_2,\ldots,{\bf c}_{k-1}$ is an odd-like (resp. even-like) binary LCD $[n,k-1]$ code and ${\bf c}_{k}\cdot {\bf c}_{i}=0$ for $1\leq i\leq k$.
\end{theorem}

\begin{proof}
Assume that there exists a basis $ {\bf c}_1, {\bf c}_2,\ldots, {\bf c}_k$ of $C$ such that the code generated by $ {\bf c}_1, {\bf c}_2,\ldots, {\bf c}_{k-1}$ is an odd-like (resp. even-like) binary LCD $[n,k-1]$ code and ${\bf c}_{k}\cdot {\bf c}_{i}=0$ for $1\leq i\leq k$. Let $G$ be a matrix whose rows are ${\bf c}_1, {\bf c}_2,\ldots, {\bf c}_k$. Then $\det(GG^T)=0$, which implies ${\rm rank}(GG^T)\leq k-1$. Since $GG^T$ contains a $(k-1)\times (k-1)$ submatrix of rank $k-1$, ${\rm rank}(GG^T)=k-1$.
By Lemma \ref{lem-2.3}, we obtain that $C$ is an odd-like (resp. even-like) binary linear code with one-dimensional hull.

Conversely, if $C$ is a binary linear code with one-dimensional hull, then there exists a basis ${\bf c}_1,{\bf c}_2,\ldots,{\bf c}_k$ of $C$ such that ${\rm Hull}(C)=\{{\bf 0},{\bf c}_k\}$. From \cite[Lemma 22]{LCD-equivalent}, the code $C'$ generated by ${\bf c}_1,{\bf c}_2,\ldots,{\bf c}_{k-1}$ is a binary LCD $[n,k-1]$ code. Moreover, it is easy to check that $C$ is odd-like (resp. even-like) if and only if $C'$ is odd-like (resp. even-like).
\end{proof}

In the following, we give a necessary condition for a binary linear code with one-dimensional hull to be even-like.

\begin{lem}\label{lem-k-odd}
If there exists an even-like binary linear $[n,k]$ code with one-dimensional hull, then $k$ is odd.
\end{lem}

\begin{proof}
Let $C$ be an even-like binary linear $[n,k]$ code with one-dimensional hull and a basis ${\bf c}_1,{\bf c}_2,\ldots,{\bf c}_k$ such that ${\rm Hull}(C)=\{{\bf0},{\bf c}_k\}$. Let $G'$ be a matrix whose rows are ${\bf c}_1,{\bf c}_2,\ldots,{\bf c}_{k-1}$.
By \cite[Lemma 22]{LCD-equivalent}, the generator matrix $G'$ generates an even-like binary LCD $[n,k-1]$ code. It follows from Theorem \ref{even-like-lCD} that $k-1$ is even. Hence $k$ is odd.
\end{proof}

\subsection{The shortened codes and the punctured codes}

Let $C$ be a binary linear $[n, k, d]$ code, and let $T$ be a set of $t$ coordinate positions in $C$. We puncture $C$ by deleting all the coordinates in $T$ in each codeword of $C$. The resulting code is still linear and has length $n - t$. We denote the punctured code by $C^T$.
Consider the set $C(T)$ of codewords which are $0$ on $T$; this set is a subcode of $C$. Puncturing $C(T)$ on $T$ gives a binary code of length $n -t$ called the code shortened on $T$ and denoted $C_T$.

\begin{lem}\label{lem-shorten-puncture}{\rm\cite[Theorem 1.5.7]{Huffman}}
Let $C$ be a binary linear $[n,k,d]$ code. Let $T$ be a set of $t$ coordinates. Then:
\begin{itemize}
  \item [(1)] $(C^\perp)_T=(C^T)^\perp$ and $(C^\perp)^T=(C_T)^\perp$, and
  \item [(2)] if $t<d$, then $C^T$ and $(C^\perp)_T$ have dimensions $k$ and $n- t- k$, respectively.
\end{itemize}
\end{lem}

\begin{lem}\label{thm-SO+Hull}
Let $C$ be a binary linear $[n,k]$ code. Let $s$ and $t$ be two integers such that $s\geq t$.
Then $C$ has $s$-dimensional hull if and only if there are $C_1$ and $C_2$ such that
$$C=C_1\oplus C_2,~C_1\subseteq C_2^\perp,$$
 where $C_1$ is a binary self-orthogonal $[n,t]$ code and $C_2$ is a binary linear $[n,k-t]$ code with $(s-t)$-dimensional hull.
\end{lem}

\begin{proof}
Assume that there is a binary self-orthogonal $[n,t]$ code $C_1$ and a binary linear $[n,k-t]$ code $C_2$ with  $(s-t)$-dimensional hull such that
$$C=C_1\oplus C_2,~C_1\subseteq C_2^\perp.$$
Then $C^\perp=C_1^\perp\cap C_2^\perp$ and $C_2\subseteq C_1^\perp$. Hence
\begin{align*}
  {\rm Hull}(C_1) & =C_1\cap C_1^\perp\subseteq C_1^\perp\cap C_2^\perp=C^\perp,  \\
  {\rm Hull}(C_2) & =C_2\cap C_2^\perp\subseteq C_1^\perp\cap C_2^\perp=C^\perp.
\end{align*}
Since $C^\perp$ is linear, ${\rm Hull}(C_1)\oplus {\rm Hull}(C_2) \subseteq C^\perp$. Obviously, ${\rm Hull}(C_1)\oplus {\rm Hull}(C_2) \subseteq C_1\oplus C_2=C$. Hence
 $$C_1\oplus {\rm Hull}(C_2)={\rm Hull}(C_1)\oplus {\rm Hull}(C_2)\subseteq C\cap C^\perp= {\rm Hull}(C).$$

 On the other hand, it can be checked that ${\rm Hull}(C)\subseteq C$ and $C\setminus (C_1\oplus {\rm Hull}(C_2))=(C_1\oplus C_2)\setminus(C_1\oplus {\rm Hull}(C_2))=C_1\oplus (C_2\setminus {\rm Hull}(C_2))$.
 For any ${\bf c}\in C_1\oplus (C_2\setminus {\rm Hull}(C_2))$, there exist ${\bf c}_1\in C_1$ and ${\bf c}_2\in C_2\setminus {\rm Hull}(C_2)$ such that
 ${\bf c}={\bf c}_1+{\bf c}_2$. Since ${\bf c}_2\in C_2\setminus {\rm Hull}(C_2)$, there exists ${\bf c}_3\in C_2$ such that ${\bf c}_2\cdot {\bf c}_3\neq 0$. Since $C_1\subseteq C_2^\perp$, ${\bf c}_1\cdot {\bf c}_3=0$. Then ${\bf c}\cdot {\bf c}_3={\bf c}_2\cdot {\bf c}_3\neq 0$.
 Hence ${\bf c}\notin {\rm Hull}(C)$ for any ${\bf c}\in C\setminus (C_1\oplus {\rm Hull}(C_2))= C_1\oplus (C_2\setminus {\rm Hull}(C_2))$. It follows that
 $${\rm Hull}(C)\subseteq C\setminus(C\setminus (C_1\oplus {\rm Hull}(C_2)))=C_1\oplus {\rm Hull}(C_2).$$
It turns out that ${\rm Hull}(C)=C_1\oplus {\rm Hull}(C_2)$ and
$$\dim({\rm Hull}(C))=\dim(C_1)+\dim({\rm Hull}(C_2))=t+s-t=s.$$
Hence, $C$ has $s$-dimensional hull.

Conversely, assume that $C$ has $s$-dimensional hull. Let $\{\alpha_1,\alpha_2,\ldots,\alpha_k\}$ be a basis of $C$ such that
$\{\alpha_1,\alpha_2,\ldots,\alpha_s\}$ is a basis of ${\rm Hull}(C)$.
Then since $t \le s$, the code $C_1$ generated by $\alpha_{1},\alpha_{2},\ldots,\alpha_{t}$ is self-orthogonal.
Let $C_2$ be a linear code generated by $\alpha_{t+1},\alpha_{t+2},\ldots,\alpha_{k}$.
Thus $C_1\subseteq C_2^\perp$ and $C=C_1\oplus C_2$.
It turns out that $C_2$ is a binary linear $[n,k-t]$ code with $(s-t)$-dimensional hull.
Otherwise, based on the discussion above, we have
$$\dim({\rm Hull}(C))=\dim(C_1)+\dim({\rm Hull}(C_2))\neq t+(s-t)=s,$$
which is a contradiction. This completes the proof.
\end{proof}

Bouyuklieva \cite{B-LCD-40} established a relation between $C$ and a shortened code of $C$. We will further subdivide this result.

\begin{prop}\label{prop-short-puncture}
Let $C$ be a binary linear $[n,k]$ code with $s$-dimensional hull. For $1\leq i\leq n$, let $C_{\{t\}}$ and $C^{\{t\}}$ be the shortened code and the punctured code of $C$ on $t$-th coordinate, respectively. Then we have the following result.
 \begin{itemize}
   \item [(1)] Assume that $s\geq 1$ and $t\in T$ for some information set $T$ of ${\rm Hull}(C)$. Then
   $$\dim\left({\rm Hull}(C^{\{t\}})\right)=\dim\left({\rm Hull}(C_{\{t\}})\right)=s-1.$$
   \item [(2)] Assume that $t\notin T$ for any information set $T$ of ${\rm Hull}(C)$. If $s\geq 1$, then
   $$s-1\leq \dim\left({\rm Hull}(C^{\{t\}})\right),\dim\left({\rm Hull}(C_{\{t\}})\right)\leq s+1.$$
   If $s=0$, then
   $$\dim\left({\rm Hull}(C^{\{t\}})\right)\leq 1~{\rm and}~\dim\left({\rm Hull}(C_{\{t\}})\right)\leq 1.$$
 \end{itemize}
\end{prop}

\begin{proof}
(1)
 Let $C$ be a binary linear $[n,k,d]$ code with $s$-dimensional hull and generator matrix $G$. Without loss of generality, we may assume that $t=1$ and
$$G=(I_k|A)=({\bf e}_{k,i}|{\bf a}_i)_{1\leq i\leq k},$$
where ${\bf e}_{k,i}$ and ${\bf a}_i$ are the $i$-th row of $I_k$ (the identity matrix) and $A$, respectively.

Since $t\in T$ for some information set of ${\rm Hull}(C)$, there exists ${\bf r}_j\in {\rm Hull}(C)$ such that $r_{j,1}=1$, where ${\bf r}_j=(r_{j,1},\ldots,r_{j,n})$.
Then we know that $\{{\bf r}_j\}\cup \{({\bf e}_{k,i}|{\bf a}_i)\}_{2\leq i\leq k}$ is a basis of $C$. By Lemma \ref{thm-SO+Hull}, the code $C'$ with the generator matrix
$({\bf e}_{k,i}|{\bf a}_i)_{2\leq i\leq k}$ is a linear $[n,k-1]$ code with $(s-1)$-dimensional hull.
By deleting the zero column of $({\bf e}_{k,i}|{\bf a}_i)_{2\leq i\leq k}$, we obtain the following matrix
$$G_{\{1\}}=({\bf e}_{k-1,i}|{\bf a}_{i+1})_{1\leq i\leq k-1},$$
 where ${\bf e}_{k-1,i}$ is the $i$-th row of $I_{k-1}$ and ${\bf a}_{i}$ is the $i$-th row of $A$ for $1\leq i\leq k-1$. This is a generator matrix of the shortened code $C_{\{1\}}$ of $C$ on the first coordinate.

Since $C'$ is a linear code with $(s-1)$-dimensional hull, $C_{\{1\}}$ be a linear $[n-1,k-1,d^*\geq d]$ code with $(s-1)$-dimensional hull.
 This completes the proof.\vspace{0.2cm}

(2) It follows from \cite[Proposition 5]{B-LCD-40} that
\begin{center}
$s-1\leq \dim\left({\rm Hull}(C_{\{t\}})\right)\leq s+1$ for $s\geq 1$ and $\dim\left({\rm Hull}(C_{\{t\}})\right)\leq 1$ for $s=0$.
\end{center}
It turns out that
\begin{center}
$s-1\leq \dim\left({\rm Hull}((C^\perp)_{\{t\}})\right)\leq s+1$ for $s\geq 1$ and $\dim\left({\rm Hull}((C^\perp)_{\{t\}})\right)\leq 1$ for $s=0$.
\end{center}
By Lemma \ref{lem-shorten-puncture}, we have
\begin{center}
$s-1\leq \dim\left({\rm Hull}(C^{\{t\}})\right)=\dim\left({\rm Hull}((C^{\{t\}})^\perp)\right)=
\dim\left({\rm Hull}((C^\perp)_{\{t\}})\right)\leq s+1 ~{\rm if}~s\geq 1,$ \\
$\dim\left({\rm Hull}(C^{\{t\}})\right)\leq 1$ if $s=0$.
\end{center}
This completes the proof.
\end{proof}

A similar result for the inverse of the punctured codes is given as follows.
\begin{prop}\label{prop-extend}
Let $C$ be a binary linear $[n,k]$ code with generator matrix $G$ and $\dim\left({\rm Hull}(C)\right)=s$. Let $C'$ be a binary linear $[n+1,k]$ code with the generator matrix $({\bf v}^T,G)$, where ${\bf v}\in \F_2^k$. Then $s-1\leq \dim\left({\rm Hull}(C')\right)\leq s+1$ for $s\geq 1$ and $\dim\left({\rm Hull}(C')\right)\leq 1$ for $s=0$.
\end{prop}

\begin{proof}
Let $s \ge 0$.
It is easy to see that $C$ is the punctured code of $C'$ on the first coordinate. Let $s'=\dim\left({\rm Hull}(C')\right)$. We claim that $s'\leq s+1$. Otherwise $\dim\left({\rm Hull}(C)\right)> s$ by Proposition \ref{prop-short-puncture}, which is a contradiction.

If $s\geq 1$, then we claim that $s'\geq s-1$. Otherwise $\dim\left({\rm Hull}(C)\right)< s$ by Proposition \ref{prop-short-puncture}, which is a contradiction. This completes the proof.
\end{proof}

\section{Linear codes with one-dimensional hull from LCD codes}

We recall that $d_{LCD}(n,k)$ is the largest minimum distance among all binary LCD $[n,k]$ codes and $d_{one}(n,k)$ is the largest minimum distance among all binary $[n,k]$ codes with one-dimensional hull for a given pair $(n, k)$.
By Proposition \ref{prop-short-puncture}, we give two upper bounds for $d_{one}(n,k)$. These two upper bounds are very helpful to determine the exact values of $d_{one}(n,k)$ since there are many known results about $d_{LCD}(n,k)$.

\begin{lem}\label{lemma-1}
Suppose that $1 \leq k \leq n-1$ and $d_{one}(n,k)\geq 2$. Then we have
 \begin{itemize}
   \item [{\rm(1)}] $d_{one}(n,k)\leq d_{LCD}(n-1,k-1)$.
   \item [{\rm(2)}] $d_{one}(n,k)\leq d_{LCD}(n-1,k)+1$.
 \end{itemize}
\end{lem}
\begin{proof}
Let $C$ be a binary linear $[n,k,d_{one}(n,k)]$ code with one-dimensional hull.
By Proposition \ref{prop-short-puncture} and Lemma \ref{lem-shorten-puncture}, there are binary LCD $[n-1,k-1,\geq d_{one}(n,k)]$ and $[n-1,k,\geq d_{one}(n,k)-1]$ codes.
Hence $d_{one}(n,k)\leq d_{LCD}(n-1,k-1)$ and $d_{one}(n,k)\leq d_{LCD}(n-1,k)+1$. This completes the proof.
\end{proof}

Next, an interesting relationship between $d_{LCD}(n+1,k)$ and $d_{one}(n,k)$ is given as follows.

\begin{lem}\label{lemma-0-1}
Suppose that $2 \leq k \leq n-1.$ Then we have $d_{one}(n,k)\leq d_{LCD}(n+1,k)$.
\end{lem}

\begin{proof}
Let ${\bf c}_1,{\bf c}_2,\ldots,{\bf c}_k$ be a basis of a binary linear $[n,k,d]$ code $C$ with one-dimensional hull such that ${\rm Hull}(C)=\{{\bf0},{\bf c}_k\}$. Let $G_1$ be a matrix whose rows are ${\bf c}_1,{\bf c}_2,\ldots,{\bf c}_{k-1}$.
According to \cite[Lemma 22]{LCD-equivalent}, $G_1G_1^T$ is nonsingular.
Let $C'$ be a binary linear $[n+1,k]$ code with the generator matrix $G'$ whose rows are $(0,{\bf c}_1),(0,{\bf c}_2),\ldots,(0,{\bf c}_{k-1}),(1,{\bf c}_k)$.
Then $C'$ has the minimum distance at least $d$, and we have
 $$G'G'^T=\left(
                            \begin{array}{cc}
                              G_1G_1^T & \begin{array}{c}
                                            0 \\
                                            \vdots \\
                                            0
                                          \end{array}
                               \\

                                0  \cdots  0& 1 \\
                            \end{array}
                          \right).
$$
Hence $G'G'^T$ is nonsingular, which implies that $C'$ is a binary LCD $[n+1,k]$ code. Thus, the result holds.
\end{proof}

\begin{prop}\label{prop-even-0-1}
If $C$ is an even-like binary LCD $[n,k,d]$ code with $d\geq 2$ and $d^\perp\geq 2$, then the shortened code of $C$ on any coordinate has one-dimensional hull.
\end{prop}

\begin{proof}
It follows from \cite[Proposition 2]{B-LCD-40} that the punctured code of $C$ on any
coordinate is again LCD. According to \cite[Lemma 2]{B-LCD-40}, exactly one of the codes $C^{\{t\}}$ and $C_{\{t\}}$ is LCD on any coordinate. By (2) of Proposition \ref{prop-short-puncture}, the shortened code of $C$ on any coordinate is either an LCD code or a linear code with one-dimensional hull. Hence the shortened code of $C$ on any coordinate has one-dimensional hull.
\end{proof}

\begin{cor}\label{cor-1}
If there is an even-like binary LCD $[n,k,d]$ code with $d\geq 2$, then there is an even-like binary linear $[n-1,k-1,\geq d]$ code with one-dimensional hull.
\end{cor}

\begin{proof}
If $d^\perp\geq 2$, then the result follows from Proposition \ref{prop-even-0-1}.
If $d^\perp=1$, then we obtain an even-like binary LCD $[n-i,k,d]$ code with the dual distance at least $2$ by deleting all zero columns of $C$. By Proposition \ref{prop-even-0-1}, there is an even-like binary $[n-i-1,k-1,d]$ code with one-dimensional hull. By adding zero-column, an even-like binary linear $[n-1,k-1,\geq d]$ code with one-dimensional hull is constructed.
\end{proof}

\begin{prop}\label{prop-odd-0-1}
Let $C$ be an odd-like binary LCD $[n,k,d]$ code with $d\geq 2$ and $d^\perp\geq 2$. If ${\bf1}\in C$, then the punctured code of $C$ on any coordinate has one-dimensional hull. If ${\bf1}\notin C$, then there exist $1\leq i,j\leq n$ such that the shortened code $C_{\{i\}}$ of $C$ on the $i$-th coordinate and the punctured code $C^{\{j\}}$ of $C$ on the $j$-th coordinate have one-dimensional hull.
\end{prop}

\begin{proof}
If $C$ contains the all-ones vector, then $C^\perp$ is even-like. By Proposition \ref{prop-even-0-1}, the shortened code of $C^\perp$ on any coordinate has one-dimensional hull. By Lemma \ref{lem-shorten-puncture}, $C^{\{t\}}=\left((C^\perp)_{\{t\}}\right)^\perp$. Hence the punctured code of $C$ on any coordinate has one-dimensional hull.

If $C$ does not contain the all-ones vector, then it follows from
\cite[Proposition 3]{B-LCD-40} that there exist $1\leq i,j\leq n$ such that the punctured code $C_{\{i\}}$ of $C$ on the $i$-th coordinate and the shortened code $C^{\{j\}}$ of $C$ on the $j$-th coordinate are LCD codes. By \cite[Lemma 2]{B-LCD-40}, exactly one of the codes $C^{\{t\}}$ and $C_{\{t\}}$ is LCD on any coordinate $t$. By (2) of Proposition \ref{prop-short-puncture}, the shortened code $C_{\{i\}}$ of $C$ on the $i$-th coordinate and the punctured code $C^{\{j\}}$ of $C$ on the $j$-th coordinate have one-dimensional hull.
\end{proof}

\begin{cor}\label{cor-2}
If there is an odd-like binary LCD $[n,k,d]$ code with $d\geq 2$, then there is a binary linear $[n-1,k,\geq d-1]$ code with one-dimensional hull.
\end{cor}

\begin{proof}
The proof is similar to that of Corollary \ref{cor-1}, the main difference is that we use Proposition \ref{prop-odd-0-1} instead of Proposition \ref{prop-even-0-1}.
\end{proof}

\begin{cor}
If $k$ is odd and $d_{LCD}(n+1,k)\geq 2$, then
\begin{center}
 $d_{one}(n,k)=d_{LCD}(n+1,k)$ or
$d_{LCD}(n+1,k)-1$.
\end{center}
\end{cor}

\begin{proof}
Let $C$ be a binary LCD $[n+1,k,d_{LCD}(n+1,k)]$ code. Since $k$ is odd, $C$ is odd-like. By Corollary \ref{cor-2}, there is a binary linear $[n , k,\geq d_{LCD}(n+1,k)-1]$ code with one-dimensional hull. Combining with Lemma \ref{lemma-0-1}, we obtain the desired result.
\end{proof}

\begin{prop}\label{prop-extend-0-1}
If there is a binary LCD $[n,k,d]$ code for odd $k$, then there is an even-like binary linear $[n+1,k,d~{\rm or}~d+1]$ code with one-dimensional hull.
\end{prop}

\begin{proof}
Let $C$ be a binary LCD $[n,k,d]$ code. Since $k$ is odd, $C$ is odd-like by Theorem~\ref{even-like-lCD}.
By Theorem \ref{odd-like-lCD}, there exists a basis
${\bf c}_1, {\bf c}_2, \ldots, {\bf c}_k$ of $C$ such that for any $i, j\in \{1, 2, \ldots , k\}$, ${\bf c}_i \cdot {\bf c}_j$ equals 1 if $i = j$ and equals 0 if $i \neq j$. Let $C'$ be a binary linear $[n+1,k]$ code with the generator matrix $G'$ whose rows are the codewords $(1,{\bf c}_1),(1,{\bf c}_2),\ldots,(1,{\bf c}_k)$.
Then we have
$$G'G'^T=J_k-I_k,$$
where $J_k$ is the all-ones matrix and $I_k$ is the $k\times k$ identity matrix. It is not difficult to calculate that $\det(G'G'^T)=0$ since $k$ is odd. Implying that $C'$ is not LCD. By Proposition \ref{prop-extend}, $C'$ is an even-like binary linear $[n+1,k]$ code with one-dimensional hull.
\end{proof}

\begin{cor}\label{cor-3.9}
Suppose that $1\leq k\leq n-1$, $k$ is odd and $d_{one}(n,k)\geq 2$. Then
\begin{center}
$d_{one}(n,k)= d_{LCD}(n-1,k)$ or $d_{LCD}(n-1,k)+1$.
\end{center}
In particular, if $d_{LCD}(n-1,k)$ is odd, then $d_{one}(n,k)= d_{LCD}(n-1,k)+1$.
\end{cor}

\begin{proof}
Let $C$ be a binary LCD $[n-1,k,d_{LCD}(n-1,k)]$ code. By Proposition \ref{prop-extend-0-1}, there is a binary linear $[n,k,d_{LCD}(n-1,k)~{\rm or}~d_{LCD}(n-1,k)+1]$ code with one-dimensional hull.
Combining with Lemma \ref{lemma-1}, we have
$$d_{LCD}(n-1,k)\leq d_{one}(n,k)\leq d_{LCD}(n-1,k)+1.$$
In particular, if $d_{LCD}(n-1,k)$ is odd. then it follows from Proposition \ref{prop-extend-0-1} that there is a binary linear $[n,k,d_{LCD}(n-1,k)+1]$ code with one-dimensional hull. Hence $d_{one}(n,k)= d_{LCD}(n-1,k)+1$.
\end{proof}

\section{Some properties of binary linear codes with one-dimensional hull}
Let $d(n,k)$ be the largest minimum distance among all binary linear $[n,k]$ codes.
It is well-known that $d(n,k)\leq d(n,k-1)$. Carlet {\em et al.} \cite{binaryLCD} proved that $d_{LCD}(n,k)\leq d_{LCD}(n,k-1)$ for any $k \ge 2$ using a new characterization of binary LCD codes, which solved the conjecture on the minimum distance of binary LCD codes proposed by Galvez {\em et al.} \cite{bound-12}. This conclusion is no longer valid for $d_{one}(n,k)$. Therefore, this is a result different from linear codes and LCD codes.

\begin{theorem}
Suppose that $2\leq k\leq n-1$. If $k$ is even or $n$ is odd, then $$d_{one}(n,k)\leq d_{one}(n,k-1).$$
\end{theorem}

\begin{proof}
Let $C$ be a binary linear $[n,k,d_{one}(n,k)]$ code with one-dimensional hull.

Assume that $k$ is even. Then $C$ is odd-like. By Theorem \ref{thm-one}, there exists a basis
${\bf c}_1, {\bf c}_2,\ldots, {\bf c}_k$ of $C$ such that the code generated by ${\bf c}_1,{\bf c}_2,\ldots,{\bf c}_{k-1}$ is an odd-like binary LCD $[n,k-1]$ code and ${\bf c}_{k}\cdot {\bf c}_{i}=0$ for $1\leq i\leq k$. Without loss of generality, we assume that ${\bf c}_1,{\bf c}_2,\ldots,{\bf c}_{k-1}$ satisfy the conditions in Theorem \ref{odd-like-lCD}. Let $C'$ be the code generated by ${\bf c}_1,{\bf c}_2,\ldots,{\bf c}_{k-2},{\bf c}_k$. By Theorems \ref{odd-like-lCD} and \ref{thm-one}, $C'$ is a binary linear $[n,k-1]$ code with one-dimensional hull and the minimum distance at least $d_{one}(n,k)$.

Assume that $n$ is odd. If $C$ is odd-like, then the result is similar to the case where $k$ is odd.
In the following, assume that $C$ is even-like. From Theorems \ref{thm-one} and \ref{even-like-lCD}, $k$ is odd and there exists a basis $ {\bf c}_1,{\bf c}'_1,\ldots,{\bf c}_{\frac{k-1}{2}},{\bf c}'_{\frac{k-1}{2}}, {\bf c}_k$ of $C$ such that the code generated by ${\bf c}_1,{\bf c}'_1,\ldots,{\bf c}_{\frac{k-1}{2}},{\bf c}'_{\frac{k-1}{2}}$ is an even-like binary LCD $[n,k-1]$ code, ${\bf c}_{k}\cdot {\bf c}_{i}={\bf c}_{k}\cdot {\bf c}'_{i}=0$ for $1\leq i\leq \frac{k-1}{2}$, and for any $i, j\in \{1, 2, \ldots , \frac{k-1}{2}\}$, the following conditions hold
(i) $ {\bf c}_i\cdot  {\bf c}_i= {\bf c}'_i\cdot  {\bf c}'_i=0$;
(ii) $ {\bf c}_i\cdot  {\bf c}'_j=0$, for $i\neq j$;
(iii) $ {\bf c}_i\cdot  {\bf c}'_i=1$;
(iv) ${\bf c}_{i,1}={\bf c}'_{i,1}$, where ${\bf c}_i=(c_{i,1},\ldots,c_{i,n})$ and ${\bf c}'_i=(c'_{i,1},\ldots,c'_{i,n})$.
Since $n$ is odd, ${\bf c}_k\neq (1,\ldots,1)$.
Without loss of generality, assume that $c_{k,1}=0$, where ${\bf c}_{k}=(c_{k,1},c_{k,2},\ldots,c_{k,n})$.

Suppose that $c_{i,1}=c'_{i,1}=0$ for $1\leq i\leq \frac{k-1}{2}$.
According to \cite[Theorem 8]{binaryLCD}, the code generated by $S_1=\{{\bf c}_1,{\bf c}'_1,\ldots, {\bf c}_{\frac{k-3}{2}},{\bf c}'_{\frac{k-3}{2}},{\bf c}_{\frac{k-1}{2}}+{\bf e}_1\}$ is a binary LCD $[n,k-2]$ code, where ${\bf e}_1=(1,0,\ldots,0)$. Let $C'$ be a binary linear code generated by $\{{\bf c}_k\}\cup S_1$. By Theorem \ref{thm-one}, $C'$ is a binary linear $[n,k-1]$ code with one-dimensional hull and the minimum distance at least $d_{one}(n,k)$.

Suppose that $c_{i,1}\neq 0$ for some $1\leq i\leq \frac{k-1}{2}$. Without loss of generality, assume that
$c_{i,1}=c'_{i,1}=1$ for $1\leq i\leq l$ and $c_{j,1}=c'_{j,1}=0$ for $l+1\leq j\leq \frac{k-1}{2}$, where $l$ is some positive integer.
According to \cite[Theorem 8]{binaryLCD}, the linear code generated by $S_2=\{{\bf c}_1+{\bf c}'_1+{\bf e}_1\}\cup\{{\bf c}_i+{\bf c}_1,{\bf c}'_i+{\bf c}_1~|~2\leq i\leq l\}\cup\{ {\bf c}_j,{\bf c}_j'~|~l+1\leq j\leq \frac{k-1}{2}\}$ is a binary LCD $[n,k-2]$ code, where ${\bf e}_1=(1,0,\ldots,0)$.
Let $C'$ be a binary linear code generated by $\{{\bf c}_k\}\cup S_2$. By Theorem \ref{thm-one}, $C'$ is a binary linear $[n,k-1]$ code with one-dimensional hull and the minimum distance at least $d_{one}(n,k)$.
This completes the proof.
\end{proof}

\begin{remark}
When $k$ is odd and $n$ is even, the above theorem may not be true. Therefore, this is a result different from linear codes and LCD codes. For example, $d_{one}(18,8)=5$, $d_{one}(18,9)=6$ (see Table 1).
\end{remark}

\begin{prop}\label{prop-4.3}
If there is an odd-like (resp. even-like) binary linear $[n,k,d]$ code with one-dimensional hull for an odd $k$, then there is an even-like (resp. odd-like) binary linear $[n+1,k,d~{\rm or}~d+1]$ code with one-dimensional hull.
\end{prop}

\begin{proof}
Let $C$ be an odd-like binary linear $[n,k,d]$ code with one-dimensional hull.
From Theorems \ref{odd-like-lCD} and \ref{thm-one}, there exists a basis ${\bf c}_1, {\bf c}_2,\ldots, {\bf c}_k$ of $C$ such that for any $i, j\in \{1, 2, \ldots , k-1\}$, ${\bf c}_i \cdot {\bf c}_j$ equals 1 if $i = j$ and equals 0 if $i \neq j$, ${\bf c}_{k}\cdot {\bf c}_{i}=0$ for $1\leq i\leq k$.
Let $C'$ be a binary linear code with the generator matrix $G'$ whose rows are $(1,{\bf c}_1),\ldots,(1,{\bf c}_{k-1}),(0,{\bf c}_k)$. Then $C'$ is an even-like binary code with the minimum distance at least $d$. According to \cite[Proposition 1]{B-LCD-40}, the linear code generated by $(1,{\bf c}_1),(1,{\bf c}_2),\\ \ldots,(1,{\bf c}_{k-1})$ is an LCD code. By Theorem \ref{thm-one}, $C'$ is an even-like binary linear $[n+1,k,d~{\rm or}~ d+1]$ code with one-dimensional hull.

Let $C$ be an even-like binary linear $[n,k,d]$ code with one-dimensional hull. By Theorems \ref{even-like-lCD} and \ref{thm-one}, there exists a basis
${\bf c}_1,{\bf c}'_1,\ldots,{\bf c}_{\frac{k-1}{2}},{\bf c}'_{\frac{k-1}{2}}, {\bf c}_k$ of $C$ such that ${\bf c}_1,{\bf c}'_1,\ldots,{\bf c}_{\frac{k-1}{2}},\\ {\bf c}'_{\frac{k-1}{2}}$ satisfy the conditions of Theorem \ref{even-like-lCD} and ${\bf c}_{k}\cdot {\bf c}_{i}={\bf c}_{k}\cdot {\bf c}'_{i}=0$ for $1\leq i\leq \frac{k-1}{2}$.
Similar to the discussion above, the code generated by $(1,{\bf c}_1),(1,{\bf c}_1'),\ldots,(1,{\bf c}_{\frac{k-1}{2}}),(1,{\bf c}'_{\frac{k-1}{2}}),(0,{\bf c}_k)$ is an odd-like binary linear $[n+1,k,d~{\rm or}~ d+1]$ code with one-dimensional hull.
\end{proof}

\begin{cor}
If $k$ is odd and $d_{one}(n-1,k)$ is odd, then $d_{one}(n,k)\geq d_{one}(n-1,k)+1$.
\end{cor}

\begin{proof}
The proof is straightforward by Proposition \ref{prop-4.3}, so we omit it here.
\end{proof}

The following propositions show some properties of the shortened and punctured codes
of binary linear codes with one-dimensional hull.

\begin{prop}\label{1-puncturn}
Let $C$ be an even-like binary linear $[n,k]$ code with ${\rm Hull}(C)=\{{\bf 0},{\bf c}_k\}$.
If $t\notin supp({\bf c}_k)$, then the punctured code $C^{\{t\}}$ of $C$ on the $t$-th coordinate is a binary linear code with one-dimensional hull.
If $t\in supp({\bf c}_k)$, then the punctured code $C^{\{t\}}$ and the shortened code $C_{\{t\}}$ of $C$ on the $t$-th coordinate are binary LCD codes.
\end{prop}

\begin{proof}
Let $C$ be an even-like binary linear $[n,k]$ code with one-dimensional hull. From Theorems \ref{even-like-lCD} and \ref{thm-one}, there exists a basis
${\bf c}_1,{\bf c}'_1,\ldots,{\bf c}_{\frac{k-1}{2}},{\bf c}'_{\frac{k-1}{2}}, {\bf c}_k$ of $C$ such that ${\bf c}_1,{\bf c}'_1,\ldots,\\{\bf c}_{\frac{k-1}{2}},{\bf c}'_{\frac{k-1}{2}}$ satisfy the conditions of Theorem \ref{even-like-lCD} and ${\bf c}_{k}\cdot {\bf c}_{i}={\bf c}_{k}\cdot {\bf c}'_{i}=0$ for $1\leq i\leq \frac{k-1}{2}$.
Let $C'$ be the code generated by ${\bf c}_1,{\bf c}'_1,\ldots,{\bf c}_{\frac{k-1}{2}},{\bf c}'_{\frac{k-1}{2}}$.
Let ${\bf c}_k'=(c_{k,1},\ldots,c_{k,t-1},c_{k,t+1},\ldots,c_{k,n})$, where ${\bf c}_k=(c_{k,1},\ldots,c_{k,n})$.
Hence $C^{\{t\}}=(C')^{\{t\}}\oplus \langle {\bf c}_k'\rangle$.

Assume that $t\notin supp({\bf c}_k)$, i.e., $c_{k,t}=0$. It follows from \cite[Proposition 2]{B-LCD-40} that the punctured code $(C')^{\{t\}}$ of $C'$ on the $t$-th coordinate is again LCD. Since $t\notin supp({\bf c}_k)$, ${\bf c}_k' \cdot {\bf c}_k'=0$ and ${\bf c}_k'\in (C')^{\{t\}}$. Hence $C^{\{t\}}=(C')^{\{t\}}\oplus \langle {\bf c}_k'\rangle$ is a binary linear $[n,k]$ code with one-dimensional hull by Lemma \ref{thm-SO+Hull}.

If $t\in supp({\bf c}_k)$, then we can obtain the desired result by (1) of Proposition \ref{prop-short-puncture}.
\end{proof}

\begin{cor}
If $k$ is odd, $n$ is even and $d_{LCD}(n,k)$ is odd, then $d_{one}(n,k)\geq d_{LCD}(n,k)$.
\end{cor}

\begin{proof}
If there exists a binary LCD $[n,k,d_{LCD}(n,k)]$ code, then it follows from Proposition \ref{prop-extend-0-1} that there is an even-like binary linear $[n+1,k,d_{LCD}(n,k)+1]$ code $C$ with one-dimensional hull. Since $n+1$ is odd, ${\bf 1}\notin C$.
By Proposition \ref{1-puncturn}, there is a binary linear $[n,k,d_{LCD}(n,k)]$ code with one-dimensional hull. So $d_{one}(n,k)\geq d_{LCD}(n,k)$.
\end{proof}

\begin{cor}
If $k$ is odd, $n$ is even and $d_{one}(n,k)$ is odd, then $d_{one}(n,k)\leq d_{LCD}(n,k)$.
\end{cor}

\begin{proof}
If there exists a binary linear $[n,k,d_{one}(n,k)]$ code with one-dimensional hull, then it follows from Proposition \ref{prop-4.3} that there is an even-like binary linear $[n+1,k,d_{one}(n,k)+1]$ code $C$ with one-dimensional hull. Since $n+1$ is odd, ${\bf 1}\notin C$.
By Proposition \ref{1-puncturn}, there is a binary LCD $[n,k,d_{one}(n,k)]$ code. Hence $d_{one}(n,k)\leq d_{LCD}(n,k)$.
\end{proof}

\begin{prop}
Let $C$ be an odd-like binary linear code with ${\rm Hull}(C)=\langle{\bf c}\rangle$ and even-like dual. If $t\notin supp({\bf c})$, then the shortened code $C_{\{t\}}$ of $C$ has one-dimensional hull.
\end{prop}

\begin{proof}
Obviously, ${\rm Hull}(C^\perp) ={\rm Hull}(C)=\{{\bf0},{\bf c}\}$. If $t\notin supp({\bf c})$, it follows from Proposition \ref{1-puncturn} that the punctured code $(C^\perp)^{\{t\}}$ of $C^\perp$ has one-dimensional hull. By Lemma \ref{lem-shorten-puncture},
$${\rm Hull}(C_{\{t\}})={\rm Hull}((C_{\{t\}})^\perp)={\rm Hull}((C^\perp)^{\{t\}}),$$
which implies that the shortened code $C_{\{t\}}$ of $C$ has one-dimensional hull.
\end{proof}

\begin{prop}\label{prop-4.9}
Let $C$ be a binary linear $[n,k]$ code with generator matrix $G$. Let $C'$ be a binary linear $[n+2,k]$ code with the generator matrix $({\bf v}^T,{\bf v}^T,G)$, where ${\bf v}\in \F_2^k$. Then $C$ has one-dimensional hull if and only if $C'$ has one-dimensional hull.
\end{prop}

\begin{proof}
It is easy to check that $GG^T=G'G'^T$. Hence the result follows.
\end{proof}

\begin{cor}
If $d_{one}(n,k)$ is odd, then $d_{one}(n+2,k)\geq d_{one}(n,k)+1$.
\end{cor}
\begin{proof}
Let $C$ be a binary linear $[n,k,d_{one}(n,k)]$ code with one-dimensional hull and generator matrix $G$. Since $d_{one}(n,k)$ is odd, the extended code $\overline{C}$ of $C$ is a binary linear $[n+1,k,d_{one}(n,k)+1]$ code. Let us assume that $\overline{C}$ has a generative matrix $({\bf v}^T,G)$. By Proposition \ref{prop-4.9}, the generator matrix $({\bf v}^T,{\bf v}^T,G)$ generates a binary linear $[n+2,k]$ code with one-dimensional hull and the minimum distance at least $d_{one}(n,k)+1$.
\end{proof}

\section{A building-up construction for binary linear codes with one-dimensional hull}

Chen \cite{ChenIT} proved that an LCD code over $\F_{2^s}$ ($s\geq 2$) is equivalent to a linear code with one-dimension hull under a weak condition. An interesting topic is to construct binary linear codes with one-dimensional hull from binary LCD codes.
Next, we introduce a complete building-up construction for linear codes with one-dimensional hull as follows.

\begin{theorem}~{\rm \cite[Theorem 1]{Kim_preprint}}\label{build-1}
Let $C$ be a binary LCD $[n,k]$ code. Let $G$ be a generator matrix for $C$.
Suppose that ${\bf x}=(x_1, x_2, \dots, x_n) \in \F_2^n$ satisfies ${\bf x} \cdot {\bf x} =1$.
Let $y_i = {\bf x} \cdot {\bf r}_i$ for $1 \le i \le k$ where ${\bf r}_i$ is the $i$-th row of $G$.
The following matrix
\[
G_1 = \left[
\begin{array}{cc|ccc}
1   & 0   & x_1 & \dots & x_n \\
\hline
y_1 & y_1 &     &  {\bf r}_1  &     \\
y_{2} & y_{2} &  & {\bf r}_2  &   \\
\vdots & \vdots & &  \vdots    &   \\
y_k & y_k &     &   {\bf r}_k &     \\
\end{array}
\right]
\]
generates a binary linear $[n+2, k+1]$ code $C_1$ with one-dimensional hull.
\end{theorem}

\begin{example}
We start from a binary LCD $[12, 2, 6]$ code. By applying Theorem~\ref{build-1}, we can construct a binary linear $[14,3,7]$ code with one-dimensional hull and the generator matrix
\[G=\left[\begin{array}{c|c}
1 0 & 1 0 0 1 1 0 0 1 0 1 1 1 \\
\hline
11 & 111 111 0000 00\\
00 & 000 111 1111 00
\end{array}
\right]\]
\end{example}

The converse of the building-up construction is also true in the following sense.

\begin{theorem}
Let $C$ be a binary linear $[n,k,d]$ code with one-dimensional hull such that $d>2$ and ${\rm Hull}(C)\neq \langle{\bf 1}\rangle$. Then $C$ can be obtained from some binary LCD $[n-2, k-1]$ code $C_0$ using the above building-up construction.
\end{theorem}

\begin{proof}
Let $G$ be a generator matrix of $C$ with one-dimensional hull. Without loss of
generality, we may assume that
\[
G=\left[ \begin{array}{c|c|c}
1 0 &  {\bf b}_1 & {\bf a}_1 \\
0 1 &  {\bf 0} & {\bf a}_2 \\
0 0 & {\bf e}_3 & {\bf a}_3 \\
\vdots & \vdots& \vdots \\
0 0 & {\bf e}_k & {\bf a}_k \\
\end{array}
\right],
\]
where ${\bf c}=(1,0,{\bf b}_1,{\bf a}_1)\in {\rm Hull}(C)$ and ${\bf e}_i$ is the $(i-2)$-th row of $I_{k-2}$ (the identity matrix). It is not difficult to check that the following matrix
\[
\left[ \begin{array}{c|c|c}
1 1 &  {\bf b}_1 & {\bf a}_1+{\bf a}_2 \\
0 0 & {\bf e}_3 & {\bf a}_3 \\
\vdots & \vdots& \vdots \\
0 0 & {\bf e}_k & {\bf a}_k \\
\end{array}
\right]
\]
generates an LCD $[n,k-1]$ code.

It suffices to prove that there exist a vector ${\bf x} = (x_1, \dots,  x_{n-2}$) and an LCD code $C_0$ of length $n-2$ whose extended code $C_1$, by Theorem~\ref{build-1}, is a code equivalent to $C$. To do that, first consider a linear code $C_0$ with the following generator matrix:
\[
G_0=\left[ \begin{array}{c|c}
  {\bf b}_1 & {\bf a}_1+{\bf a}_2 \\
 {\bf e}_3 & {\bf a}_3 \\
 \vdots& \vdots \\
 {\bf e}_k & {\bf a}_k \\
\end{array}
\right],
\]
which is an LCD $[n-2,k-1]$ code by \cite[Proposition 4]{B-LCD-40}.

Using the row ${\bf x} = ({\bf b}_1 | {\bf a}_1)$ of length $n-2$ and $G_0$, we get a
generator matrix $G_1$ of a linear $[n,k]$ code $C_1$ by Theorem~\ref{build-1}, in this case, ${\mbox{wt}}({\bf x})$ is odd.
\[
G_1= \left[ \begin{array}{c|c|c}
1 0 &  {\bf b}_1 & {\bf a}_1 \\
1 1 &  {\bf b}_1 & {\bf a}_1+{\bf a}_2 \\
0 0 & {\bf e}_3 & {\bf a}_3 \\
\vdots & \vdots& \vdots \\
0 0 & {\bf e}_k & {\bf a}_k \\
\end{array}
\right]
\sim
\left[ \begin{array}{c|c|c}
1 0 &  {\bf b}_1 & {\bf a}_1 \\
0 1 &  {\bf 0} & {\bf a}_2 \\
0 0 & {\bf e}_3 & {\bf a}_3 \\
\vdots & \vdots& \vdots \\
0 0 & {\bf e}_k & {\bf a}_k \\
\end{array}
\right]
=G.
\]
Thus the given code $C$ is
equivalent to $C_1$, as desired. This completes the proof.
\end{proof}

\begin{theorem}
Let $C$ be a binary linear $[n,k]$ code $C$ with ${\rm Hull}(C)=\langle{\bf 1}\rangle$. Then $n$ is even, $k$ is odd and there exists an even-like binary LCD $[n, k-1]$ code $C_0$ such that
$C=C_0\oplus \langle{\bf 1}\rangle$.
\end{theorem}

\begin{proof}
Since ${\rm Hull}(C)=\langle{\bf 1}\rangle$, $C$ and $C^\perp$ are even-like, which implies that $n$ is even. Note that $k$ is odd by Lemma~\ref{lem-k-odd}.
Let ${\bf c}_1,{\bf c}_2,\ldots,{\bf c}_{k-1},{\bf 1}$ be a basis of $C$. Then it follows from \cite[Lemma 22]{LCD-equivalent} that the code $C_0$ generated by ${\bf c}_1,{\bf c}_2,\ldots,{\bf c}_{k-1}$ is an even-like binary LCD $[n,k-1]$ code. Hence $C=C_0\oplus \langle{\bf 1}\rangle$. This completes the proof.
\end{proof}

Using Theorem~\ref{build-1}, we can obtain the following corollary.

\begin{cor}\label{Methods}
 Let $C$ be a binary LCD $[n,k]$ code with generator matrix $G$. Suppose that ${\bf x}\in C^\perp$ and ${\rm wt}({\bf x})$ is odd. Then the following matrix
$$
\begin{bmatrix}
1 & {\bf x} \\
 {\rm{\bf 0}} & G
\end{bmatrix}
$$
generates a binary linear $[n+1,k+1]$ code with one-dimensional hull.
\end{cor}

\begin{proof}
The code $C_3$ constructed from Theorem~\ref{build-1} is a binary linear $[n+2, k+1]$ code with one-dimensional hull. Since ${\bf x}\in C^\perp$, $y_i=0$ for $i=1, \dots, k$. Therefore, by puncturing $C_3$ on the second coordinate, we obtain the matrix
$$
\begin{bmatrix}
1 & {\bf x} \\
 {\rm{\bf 0}} & G
\end{bmatrix},
$$
which also generates a binary linear $[n+1, k+1]$ code with one-dimensional hull.
\end{proof}

\begin{example}\label{example-1}
We start from a binary LCD $[13, 5, 5]$ code. By applying Corollary~\ref{Methods} we can construct a binary linear $[14,6,5]$ code $C$ with one-dimensional hull and the generator matrix
\[G=\left[\begin{array}{c|c}
1  & 1 0 1 1 0 1 0 0 0 1 0 1 1 \\
\hline
0  & 1 0 0 0 0 1 1 0 1 0 1 1 1\\
0  & 0 1 0 0 0 1 1 1 0 0 0 1 0 \\
0  & 0 0 1 0 0 1 0 0 0 1 1 1 0\\
0  & 0 0 0 1 0 0 0 1 1 1 0 1 1\\
0  & 0 0 0 0 1 0 1 1 1 1 1 0 1
\end{array}
\right].\]
\end{example}

\begin{theorem}
Any binary linear $[n,k,d]$ code with one-dimensional hull can be obtained from some binary LCD $[n-1,k-1,\geq d]$ code by the construction of Corollary \ref{Methods}.
\end{theorem}

\begin{proof}
Let $C$ be a binary linear $[n,k,d]$ code with one-dimensional hull.
By Proposition \ref{prop-short-puncture}, there is at least one coordinate position $i$ such that the shortened code $C_{\{i\}}$ of $C$ on the $i$-th coordinate is a binary LCD $[n-1,k-1,\geq d]$ code. Without loss of generality, we consider that $i=1$.
Assume that $C_{\{1\}}$ has the generator matrix $G_1$. Then $C$ has the generator matrix
$$G=\begin{bmatrix}
1 & {\bf x'} \\
\textbf 0 & G_1
\end{bmatrix}$$
for some ${\bf x'}=(x'_1,\ldots,x'_{n-1})\in \F_2^{n-1}$.
Since $C_{\{1\}}$ is a binary LCD code, $\F_2^{n-1}=C_{\{1\}}\oplus (C_{\{1\}})^\perp$.
So there are ${\bf x}=(x_1,\ldots,x_{n-1})\in (C_{\{1\}})^\perp$ and ${\bf y}=(y_1,\ldots,y_{n-1})\in C_{\{1\}}$ such that ${\bf x'}={\bf x}+{\bf y}$.
Hence the following matrix
 $$G_0=\begin{bmatrix}
1 & {\bf x} \\
\textbf 0 & G_1
\end{bmatrix}$$
 is also the generator matrix of $C$.
It turns out that ${\rm wt}({\bf x})$ is odd, otherwise $C\cap C^\perp=\{{\bf0}\}$, which is a contradiction.
This completes the proof.
\end{proof}

\begin{remark}
If we would like to obtain all binary $[n, k, d]$ linear codes with one-dimensional hull, then we can start from all binary LCD $[n-1, k-1, \geq d]$ codes. This theorem may be very useful in classification.
\end{remark}

\begin{example}\label{example-5.9}
According to \cite{HS-BLCD-1-16}, there exist a unique inequivalent binary LCD $[15,7,5]$ code. By applying Corollary~\ref{Methods} we cannot construct a binary linear $[16,8,5]$ code with one-dimensional hull. So $d_{one}(16,8)\leq 4.$
\end{example}

Harada and Saito \cite{HS-BLCD-1-16} gave a complete classification of optimal binary LCD $[n, k]$ codes for $1\leq k\leq n\leq 16$.
Bouyuklieva \cite{B-LCD-40} gave a partial classification of optimal binary LCD $[n, k]$ codes for $1\leq k\leq n\leq 40$. A complete classification of optimal binary LCD $[n, 3]$ codes was given in \cite{AHS-BLCD,HS-BLCD-1-16}. Applying Corollary~\ref{Methods} to these LCD codes, we have the following proposition.

\begin{prop}
There are no binary linear $[16,8,5]$, $[16,10,4]$, $[17,9,5]$, $[18,8,6]$, $[20,4,10]$, $[20,8,7]$, $[20,10,6]$, $[22,4,11]$, $[22,8,8]$, $[23,6,10]$, $[24,4,12]$, $[25,6,11]$, $[26,8,10]$, $[27,6,12]$, $[28,8,11]$, $[29,6,13]$ codes with one-dimensional hull.
\end{prop}

\begin{proof}
We start from all binary LCD $[15,7,5]$, $[15,9,4]$ and $[16,8,5]$ codes (see \cite{HS-BLCD-1-16}). By applying Corollary~\ref{Methods}, we cannot construct binary linear $[16,8,5]$, $[16,10,4]$ and $[17,9,5]$ code with one-dimensional hull.

We start from all binary LCD $[n,k,d]$ codes, where $(n,k,d)\in\{(17,7,6),(19,7,7),(19,\\9,6),(21,7,8),(22,5,10),(24,5,11),(25,7,10),(26,5,12),
(27,7,11),(28,5,13)\}$ (see \cite{B-LCD-40}).\\By applying Corollary~\ref{Methods}, we cannot construct a binary linear $[n+1,k+1,d]$ code with one-dimensional hull.

We start from all binary LCD $[n,k,d]$ codes, where $(n,k,d)\in\{(19,3,10),(21,3,11),\\(23,3,12)\}$ (see \cite{AHS-BLCD,HS-BLCD-1-16}). By applying Corollary~\ref{Methods}, we cannot construct a binary linear $[n+1,k+1,d]$ code with one-dimensional hull. This completes the proof.
\end{proof}

\begin{prop}
There is no binary linear $[23,14,5]$ code with one-dimensional hull.
\end{prop}

\begin{proof}
There exists a unique inequivalent binary linear $[23,14,5]$ code \cite{[23-14-5]}, which has $3$-dimensional hull by MAGMA \cite{magma}.
\end{proof}

\begin{remark}
For fixed $n$ and $k$, there are two upper bounds on $d_{one}(n,k)$:
\begin{center}
$d_{one}(n,k)\leq d_{LCD}(n-1,k-1)$ (Lemma \ref{lemma-1}) and $d_{one}(n,k)\leq d(n,k).$
\end{center}
For the upper bound of $d_{LCD}(n,k)$, we refer to \cite{bound-12,HS-BLCD-1-16,AH-BLCD-17-24,2-LCD-30,2-LCD-40,234-lcD,
B-LCD-40,LCD-code-Li,AHS-BLCD,ter-11-19}.
For the upper bound of $d(n,k)$, we refer to \cite{codetables}.
\end{remark}

\begin{center}
\setlength\tabcolsep{6.5pt}
\begin{tabular}{cccccccccccccccc}
\multicolumn{16}{c}{{\rm Table 1: $d_{one}(n,k)$, where $14\leq n\leq 30$, $1\leq k\leq 15$}}\\
\hline
   $n/k$& 1&2&3&4&5&6&7&8 &9 &10 &11&12&13&14&15\\
    \hline\hline
   14&14& 8&7&6 &6 &5 &4 &4 &3 &2 &2 &1 &2 & & \\
   15&14& 9&8&7 &6 &5 &5 &4 &4 &3 &2 &2 &2 &1 &   \\
   16&16& 9&8&7 &6 &6 &6 &4 &4 &3 &3 &2 &2 &1 &2   \\
   17&16& 11&9&8 &7 &6 &6 &5 &4 &4 &4 &3 &2 &2 &2  \\
   18&18& 11&10&8 &8 &7 &6 &5 &6 &4 &4 &3 &3 &2 &2 \\
   19&18& 12&10&9 &8 &7 &7 &6 &6 &5 &4 &4 &4 &3 &2 \\
   20&20& 12&10&9 &9 &8 &8 &6 &6 &5 &5 &4 &4 &3 &3 \\
   21&20& 13&11&10 &10 &8 &8 &7 &6 &6 &6 &5 &4 &4 &4   \\
   22&22& 13&12&10 &10 &9 &8 &7 &7 &6 &6 &5 &5 &4 &4  \\
   23&22& 15&12&11 &10 &9 &9 &8 &8 &7 &6 &6 &6 &4 &4  \\
   24&24& 15&13&11 &11 &10 &10 &8 &8 &7 &7 &6 &6 &5 &4  \\
   25&24& 16&14&12 &12 &10 &10 &9 &8 &8 &8 &7 &6 &5-6 &5  \\
   26&26& 16&14&13 &12 &11 &10 &9 &9 &8 &8 &7 &7 &6 &6   \\
   27&26& 17&14&13 &12 &11 &11 &10 &10 &8-9 &8 &8 &8 &7 &6 \\
   28&28& 17&15&14 &13 &12 &12 &10 &10&9 &8 &8 &8 &7 &6  \\
   29&28& 19&16&14 &14 &12 &12 &11 &10&9-10 &9 &8 &8 &8 &6-7  \\
   30&30& 19&16&15 &14 &13 &12 &11-12 &11&10 &10 &9 &8&8 &7-8  \\
    \hline\hline
    $n/k$  & 16& 17& 18&19&20&21&22&23&24&25 &26 &27 &28&29&\\
    \hline\hline
    17 &1& &&& & & & &&& & && \\
   18 &1& 2& &&& & & & &&& & &\\
   19 &2& 2& 1& && & & & &&& & & \\
   20&2& 2&1& 2&& & & & &&& & & \\
   21&3& 2&2& 2&1& & & & &&& & & \\
   22&3&3& 2& 2& 1& 2& & & &&& & & \\
   23&4&4 &3& 2& 2& 2&1 & & &&& & & \\
   24&4&4& 3& 3& 2& 2&1 &2 & &&& & & \\
   25&4& 4&4& 4& 3&2 &2 &2 &1 &&& & & \\
   26&5& 4&4& 4& 3&3 &2 &2 &1 &2&& & & \\
   27&5-6&5& 4& 4& 4&4 &3 &2 &2 &2&1& & & \\
   28&6&6& 5& 4& 4&4 &3 &2 &2 &2&1&2 & & \\
   29&6&6& 5-6& 5& 4&4 &4 &3 &2 &2&2&2 &1 & \\
   30&6&6 &6& 6& 5&4 &4 &4 &3 &2&2& 2& 1&2 \\
    \hline
\end{tabular}
\end{center}

\begin{remark}
 The value in Table 1 denotes the minimum distance of an optimal binary linear $[n,k]$ code with one-dimensional hull by our method except for the binary linear codes with one-dimensional hull in the Magma database. All computations have been done by MAGMA \cite{magma}. To save the space, the codes in Table 1 can be obtained from one of the authors' website, namely,\\
  {\tt https://cicagolab.sogang.ac.kr/cicagolab/2660.html}.
\end{remark}

\section{Optimal binary linear codes with one-dimensional hull}
In this section, we characterize the minimum distances of optimal binary linear $[n,k]$ and $[n,n-k]$ codes with one-dimensional hull for $k\leq 5$.

\subsection{Optimal binary linear $[n,1]$ and $[n,n-1]$ codes with one-dimensional hull}
In this subsection, we study the exact values of $d_{one}(n,1)$ and $d_{one}(n,n-1)$.
\begin{theorem}
If $n$ is odd, then $d_{one}(n,1)=n-1$ and $d_{one}(n,n-1)=1$.
If $n$ is even, then $d_{one}(n,1)=n$ and $d_{one}(n,n-1)=2$.
\end{theorem}

\begin{proof}
By the Griesmer bound, we have $d_{one}(n,1)\leq n$ and $d_{one}(n,n-1)\leq 2$.
Assume that $n$ is odd. The repetition $[n,1,n]$ code is not a linear code with one-dimensional hull. The code $C$ generated by $[0 1 1 \ldots 1]$ is a linear code with one-dimensional hull. So $d_{one}(n,1)=n-1$.
The dual code $C^\perp$ of $C$ is a linear code with one-dimensional hull and the minimum weight $1$. If $d_{one}(n,n-1)=2$, then the corresponding code $C'$ is the even
$[n, n- 1, 2]$ code. The dual of $C'$ is the repetition $[n,1,n]$ code, which is not a linear code with one-dimensional hull. Thus $d_{one}(n,n-1)=1.$

Assume that $n$ is even. The repetition $[n,1,n]$ code and its dual code are linear codes with one-dimensional hull. Hence $d_{one}(n,1)=n$ and $d_{one}(n,n-1)=2$.
\end{proof}

\subsection{Optimal binary linear $[n,2]$ and $[n,n-2]$ codes with one-dimensional hull}

Mankean and Jitman \cite{binary-hull-2} determined the exact value of $d_{one}(n,2)$.
\begin{theorem}{\rm\cite{binary-hull-2}}
Let $n>2$ be an integer. Then we have
$$d_{one}(n,2)= \left\{
\begin{array}{ll}
 \left\lfloor \frac{2n}{3}\right\rfloor,~ &{\rm for}~ n\equiv 1,5~({\rm mod}~6)\vspace{1ex},\\
 \left\lfloor \frac{2n}{3}\right\rfloor-1,~& {\rm for}~ n\equiv 0,2,3,4~({\rm mod}~6).
\end{array}
\right.$$
\end{theorem}

Next, we consider the exact value of $d_{one}(n,k)$ for $k=n-2$.

\begin{theorem}
Let $n>2$ be an integer. Then we have
$$d_{one}(n,n-2)= \left\{
\begin{array}{ll}
 2,~ &{\rm if}~ n~{\rm is~odd},\\
 1,~& {\rm if}~ n~{\rm is~even}.
\end{array}
\right.$$
\end{theorem}

\begin{proof}
By the Griesmer bound, $d_{one}(n,n-2)\leq 2.$ Hence $d_{one}(n,n-2)= 2$ or $1.$ Let $x,y,z,s$ be four integers.
Consider the code $C$ of length $n$ with the parity-check matrix
$$H= \left[ \begin{array}{c|c|c|c}
 1\ldots 1 & 1 \ldots 1\ & 0 \ldots 0 &~ 0 \ldots 0\\
 \undermat{x}{0 \ldots 0}  & \undermat{y}{1 \ldots  1} & \undermat{z}{1 \ldots 1}
 &~ \undermat{s}{0 \ldots 0}
 \end{array} \right].$$\\
It is easy to see that the code $C$ has minimum diatance $2$ (resp. 1) if and only if $s=0$ (resp. $s>0$).
Let $n$ be an odd integer, i.e., $n=2m+1$ for some positive integer $m$. If $x=m,y=0,z=m+1,s=0$, then the code $C$ is a binary linear $[2m+1,2m-1,2]$ code with one-dimensional hull. Therefore, $d_{one}(n,n-2)=2$ if $n$ is odd.

Let $n$ be an even integer. Assume that $s=0$.
\begin{itemize}
  \item If $x$ is odd, then $y+z$ is odd. Whether $y$ is odd or even, $C$ is an LCD code.
  \item If $x$ is even, then $y+z$ is even.
\begin{itemize}
  \item If $y$ is odd, then it is not difficult to check that $C$ is an LCD code.
  \item If $y$ is even, then it is not difficult to check that $C^\perp\subset C$.
\end{itemize}
\end{itemize}
 This implies that $s>0$ when $C$ is a binary linear code with one-dimensional hull. Therefore, $d_{one}(n,n-2)=1$ if $n$ is even.
\end{proof}

\subsection{Optimal binary linear $[n,3]$ and $[n,n-3]$ codes with one-dimensional hull}

Assume that $S_k$ is a matrix whose columns are all nonzero vectors in $\F_2^k$.
It is well-known that $S_k$ generates a binary simplex code, which is a one-weight self-orthogonal $[2^k-1,k,2^{k-1}]$ Griesmer code for $k\geq 3$ (see \cite{Huffman}).

\begin{lem}\label{lemma-k-3}
Assume that $S_k$ is a matrix whose columns are all nonzero vectors in $\F_2^k$ for $k\geq 3$. Let $C$ be a binary linear $[n,k,d]$ code with generator matrix $G$. Then $C$ has one-dimensional hull if and only if
$C'$ with the following matrix
$$G'=[\underbrace{S_k|\cdots|S_k}_m|G]$$
is a binary linear $[m(2^k-1)+n,k,m2^{k-1}+d]$ code with one-dimensional hull.
\end{lem}

\begin{proof}
It is well-known that $S_k$ generates a binary simplex code, which is a one-weight self-orthogonal $[2^k-1,k,2^{k-1}]$ Griesmer code. So
$$G'G'^T=GG^T.$$
Therefore, $C$ has one-dimensional hull if and only if $C'$ has one-dimensional hull.
Since $C$ has the minimum distance $d$, $C'$ has the minimum distance at least $d+2^{k-1}m$. Since the simplex code is a one-weight code, there is at least a codeword of weight $d+2^{k-1}m$ in $C'$. The converse is also true. This completes the proof.
\end{proof}

Let $h_{k,i}$ be the $i$-th column of the matrix $S_k$. Let $G_{k}({\bf m})$ be a $k\times \sum_{i=1}^{2^k-1}m_i$ matrix which consists of $m_i$ columns $h_{k,i}$ for each $i$ as follows:
$$G_{k}({\bf m})=[\underbrace{h_{k,1},\ldots,h_{k,1}}_{m_1},\ldots,
\underbrace{h_{k,2^k-1},\ldots,h_{k,2^k-1}}_{m_{2^k-1}}],$$
where ${\bf m}=(m_1,\ldots,m_{2^k-1})$ and $m_i$ is a nonnegative integer. For a binary linear $[n,k,d]$ code with $d(C^{\perp})\geq 2$, there exists a vector
${\bf m}=(m_1,\ldots,m_{2^k-1})$ such that $C$ is equivalent to the code $C_k({\bf m})$ with the generator matrix $G_k({\bf m})$.

\begin{prop}\label{prop-rank}
Let ${\bf m}=(m_1,m_2,\ldots,m_{2^k-1})$ and $m=\min\{m_1,m_2,\ldots,m_{2^k-1}\}$. Let $C$ be a binary linear $[n,k,d]$ code with the generator matrix $G_k({\bf m})$. Let $C'$ be a binary linear code with the generator matrix $G_k({\bf m'})$, where ${\bf m'}=(m_1-m,m_2-m,\ldots,m_{2^k-1}-m)$. If $d> m2^{k-1}$, then $C'$ is a binary linear $[n-m(2^k-1),k,d-m2^{k-1}]$ code.
\end{prop}

\begin{proof}
We just verify that $C'$ has $2^k$ codewords, i.e., ${\rm rank}(G_k({\bf m'}))=k$. Without loss of generality, let
$$G_{k}({\bf m})=[\underbrace{S_k,\ldots,S_k}_{m},G_{k}({\bf m'})].$$
Assume that ${\rm rank}(G_k({\bf m'}))<k$. Since $S_k$ generates a one-weight code, we obtain $d=m2^{k-1}$, which is a contradiction. This completes the proof.
\end{proof}

The following is an interesting and useful result proposed by Araya et al.\cite{AHS-BLCD}.

\begin{lem}{\rm\cite{AHS-BLCD}}\label{lemma-mi}
Suppose that $(q,k_0)=(2,3)$ and $k\geq k_0$. If the code
$C_{k}({\bf m})$ has minimum weight at least $d$, then
$$2d-n\leq m_i\leq n-\frac{2^{k-1}-1}{2^{k-2}}d,$$
for each $i\in \{1,2,\ldots,2^k-1\}$, where ${\bf m}=(m_1,\ldots,m_{2^k-1})$ and $n=\sum_{i=1}^{2^k-1}m_i$.
\end{lem}

\begin{prop}\label{prop-7m+6}
There is no binary linear $[7m+6,3,4m+3]$ code with one-dimensional hull for $m\geq 0$.
\end{prop}

\begin{proof}
Suppose that $C$ is a binary linear $[7m+6,3,4m+3]$ code with one-dimensional hull. Then $C$ is a Griesmer code and $d(C^\perp)\geq 2$.
Hence there is a vector ${\bf m}=(m_1,\ldots,m_{7})$ such that $C$ is equivalent to $C_3({\bf m})$.
By Lemma \ref{lemma-mi}, $m_i\geq m$. Let ${\bf m'}=(m_1-m,\ldots,m_{7}-m)$.
By Proposition \ref{prop-rank} and Lemma \ref{lemma-k-3}, the code $C_3({\bf m'})$ is a binary linear $[6,3,3]$ code with one-dimensional hull, which contradicts $d_{one}(6,3)=2$ (see \cite[Table 1]{Kim_preprint}).
Hence $d_{one}(7m+6,3)\leq 4m+2$.
\end{proof}

By the Griesmer bound and some known results, we obtain the following theorem.
\begin{theorem}
Let $n>3$ be an integer. Then we have
$$d_{one}(n,3)= \left\{
\begin{array}{ll}
 \left\lfloor \frac{4n}{7}\right\rfloor,~ &{\rm for}~ n\equiv 1,3,4,5~({\rm mod}~7)\vspace{1ex},\\
 \left\lfloor \frac{4n}{7}\right\rfloor-1,~& {\rm for}~ n\equiv 0,2,6~({\rm mod}~7).
\end{array}
\right.$$
\end{theorem}
\begin{proof}
By the Griesmer bound, we have
$$d_{one}(n,3)\leq \left\{
\begin{array}{ll}
  \left\lfloor \frac{4n}{7}\right\rfloor, & {\rm if}\ n\equiv 0,1,3,4,5,6~({\rm mod}~ 7)\vspace{1ex}, \\
 \left\lfloor \frac{4n}{7}\right\rfloor-1, & {\rm if}\ n\equiv 2~({\rm mod}~ 7).
\end{array}
\right.$$

(i) Assume that $n\equiv 4~({\rm mod}~7)$, i.e., $n=7m+4$ for some integer $m$.
Applying Lemma \ref{lemma-k-3} to the binary linear $[4,3,2]$ code with one-dimensional hull (see \cite[Table 1]{Kim_preprint}), we have
$$d_{one}(7m+4,3)\geq 4m+2=\left\lfloor \frac{4(7m+4)}{7}\right\rfloor.$$
Combining with the Griesmer bound, we have $d_{one}(n,3)=\left\lfloor \frac{4n}{7}\right\rfloor$ for $n\equiv 4~({\rm mod}~7)$.
A similar argument works for $n\equiv 1,3,5~({\rm mod}~7)$.

(ii) Assume that $n\equiv 0~({\rm mod}~7)$, i.e., $n=7m$ for some integer $m$.
By Lemma \ref{lemma-0-1} and \cite[Theorem 5.1]{HS-BLCD-1-16},
$$d_{one}(7m,3)\leq d_{LCD}(7m+1,3)=\left\lfloor \frac{4(7m+1)}{7}\right\rfloor-1=4m-1=\left\lfloor \frac{4\times 7m}{7}\right\rfloor-1.$$
On the other hand, applying Lemma \ref{lemma-k-3} to the binary linear $[7,3,3]$ code with one-dimensional hull (see \cite[Table 1]{Kim_preprint}), we have
$$d_{one}(7m,3)\geq 4m-1=\left\lfloor \frac{4\times 7m}{7}\right\rfloor-1.$$
This implies that $d_{one}(n,3)=\left\lfloor \frac{4n}{7}\right\rfloor-1$ for $n\equiv 0~({\rm mod}~7)$.

(iii) Assume that $n\equiv 2~({\rm mod}~7)$, i.e., $n=7m+2$ for some integer $m$.
Applying Lemma \ref{lemma-k-3} to the binary linear $[9,3,4]$ code with one-dimensional hull (see \cite[Table 1]{Kim_preprint}), we have
$$d_{one}(7m+2,3)\geq 4m=\left\lfloor \frac{4(7m+2)}{7}\right\rfloor-1.$$
Combining with the Griesmer bound, we have $d_{one}(n,3)=\left\lfloor \frac{4n}{7}\right\rfloor-1$ for $n\equiv 2~({\rm mod}~7)$.

(iv) Assume that $n\equiv 6~({\rm mod}~7)$, i.e., $n=7m+6$ for some integer $m$.
Applying Lemma \ref{lemma-k-3} to the binary linear $[6,3,2]$ code with one-dimensional hull (see \cite[Table 1]{Kim_preprint}), we have
$$d_{one}(7m+6,3)\geq 4m+2=\left\lfloor \frac{4(7m+6)}{7}\right\rfloor-1.$$
Combining with Proposition \ref{prop-7m+6}, we have $d_{one}(n,3)=\left\lfloor \frac{4n}{7}\right\rfloor-1$ for $n\equiv 6~({\rm mod}~7)$.
\end{proof}

\begin{lem}\label{lemma-d(n-k)}
Let $k\geq 3$ and $n\geq 2^k$. Then $d_{one}(n,n-k)=2$.
\end{lem}

\begin{proof}
If $d(n,n-k)\geq 3$, then it follows from the sphere-packing bound that
$$2^{n-k}\leq \frac{2^n}{1+n}, ~i.e.~1+n\leq 2^k,$$
which contradicts $n\geq 2^k$.
Hence $d_{one}(n,n-k)\leq d(n,n-k)\leq 2$.
Consider the code $C$ with the following matrix
$$G=\left[
       \begin{array}{cccccccccc}
         1 & 0 & \ldots & 0 & 1 & 1 & 1 &0&\ldots&0 \\
         0 & 1 & \ldots & 0 & 1 & 1 & 0 &0&\ldots&0\\
         \vdots & \vdots & \ddots & \vdots & \vdots & \vdots & \vdots& \vdots & \ddots & \vdots  \\
         0 & 0 & \cdots & 1 & 1 & 1 & 0&0&\ldots&0
       \end{array}
     \right]_{(n-k)\times n}.
$$
Then $C$ has parameters $[n,n-k,2]$ and
$$GG^T=\left[
       \begin{array}{cccc}
         0 & 0 & \ldots & 0  \\
         0 & 1 & \ldots & 0 \\
         \vdots & \vdots & \ddots & \vdots \\
         0 & 0 & \cdots & 1
       \end{array}
     \right]_{(n-k)\times (n-k)}.$$
     It turns out that ${\rm rank}(GG^T)=n-k-1$. This implies that $C$ has one-dimensional hull.
     Hence $d_{one}(n,n-k)=2$. This completes the proof.
\end{proof}

Next, we consider the exact value of $d_{one}(n,k)$ for $k=n-3$.
\begin{theorem}
Let $n\geq 4$ be an integer. Then
$$d_{one}(n,n-3)=\left\{\begin{array}{ll}
                          4, & {\rm if}~n=4, \\
                          3, & {\rm if}~n=5, \\
                          2, & {\rm if}~n\geq 6.
                        \end{array}
 \right.$$
\end{theorem}

\begin{proof}
From \cite[Table 1]{Kim_preprint}, $d_{one}(4,1)=4$, $d_{one}(5,2)=3$ and $d_{one}(6,3)=d_{one}(7,4)=2.$
It follows from Lemma \ref{lemma-d(n-k)} that $d_{one}(n,n-3)=2$ for $n\geq 8$. This completes the proof.
\end{proof}

\subsection{Optimal binary linear $[n,4]$ and $[n,n-4]$ codes with one-dimensional hull}

In this subsection, we study the exact values of $d_{one}(n,4)$ and $d_{one}(n,n-4)$.
\begin{prop}\label{prop-15m+}
There is no binary linear $\left[n,4,\left\lfloor \frac{8n}{15}\right\rfloor\right]$ code with one-dimensional hull for $n\equiv 0,1,5,7,8,9,12,14~({\rm mod}~ 15)$ and $n\geq 7$.
\end{prop}

\begin{proof}
Assume that $n\equiv 0,5,7,8,12,14~({\rm mod}~ 15)$. If there is a binary linear $\left[n,4,\left\lfloor \frac{8n}{15}\right\rfloor\right]$ code $C$ with one-dimensional hull, then it can be checked that $C$ is a Griesmer code. It turns out that $d(C^\perp)\geq 2$.

Assume that $n\equiv 7~({\rm mod}~ 15)$, i.e., $n=15m+7$ for some integer $m$.
If $C$ is a binary linear $[15m+7,4,8m+3]$ code with one-dimensional hull for $m\geq 2$, then $d(C^\perp)\geq 2$ and there is a vector ${\bf m}=(m_1,\ldots,m_{15})$ such that $C$ is equivalent to $C_4({\bf m})$.
By Lemma \ref{lemma-mi}, we have $m_i\geq m-1$. Let ${\bf m'}=(m_1-m+1,\ldots,m_{15}-m+1)$.
Combining Proposition \ref{prop-rank} and Lemma \ref{lemma-k-3}, the code $C_4({\bf m'})$ is a binary linear $[22,4,11]$ code with one-dimensional hull, which contradicts that $d_{one}(22,4)=10$ (see Table 1).
Hence there is no binary linear $\left[n,4,\left\lfloor \frac{8n}{15}\right\rfloor\right]$ code with one-dimensional hull for $n\equiv 7~({\rm mod}~ 15)$. A similar argument works for $n\equiv 0,5,8,12,14~({\rm mod}~15)$.

Assume that $n\equiv 1~({\rm mod}~ 15)$, i.e., $n=15m+1$ for some integer $m$. Suppose that $C$ is a binary linear $[15m+1,4,8m]$ code with one-dimensional hull for $m\geq 1$. If $d(C^\perp)=1$, then there is a binary linear $[15m,4,8m]$ code with one-dimensional hull, which contradicts that $d_{one}(15m,4)<8m$. Hence $d(C^\perp)\geq 2$.
Then there is a vector ${\bf m}=(m_1,\ldots,m_{15})$ such that $C$ is equivalent to $C_4({\bf m})$. By Lemma \ref{lemma-mi}, we have $m_i\geq m-1$. Let ${\bf m'}=(m_1-m+1,\ldots,m_{15}-m+1)$.
Combining Proposition \ref{prop-rank} and Lemma \ref{lemma-k-3}, the code $C_4({\bf m'})$ is a binary linear $[16,4,8]$ code with one-dimensional hull, which contradicts that $d_{one}(16,4)=7$ (see Table 1).
Hence there is no binary linear $\left[n,4,\left\lfloor \frac{8n}{15}\right\rfloor\right]$ code with one-dimensional hull for $n\equiv 1~({\rm mod}~ 15)$. A similar argument works for $n\equiv 9~({\rm mod}~15)$.
\end{proof}

\begin{theorem}
Let $n\geq 7$ be an integer. Then we have
$$d_{one}(n,4)= \left\{
\begin{array}{ll}
  \left\lfloor \frac{8n}{15}\right\rfloor, & {\rm if}\ n\equiv 11,13~({\rm mod}~ 15)\vspace{1ex}, \\
 \left\lfloor \frac{8n}{15}\right\rfloor-1, & {\rm if}\ n\equiv 0,1,2,3,4,5,6,7,8,9,10,12,14~({\rm mod}~ 15).
\end{array}
\right.$$
\end{theorem}

\begin{proof}
By the Griesmer bound, we have
$$d_{one}(n,4)\leq \left\{
\begin{array}{ll}
  \left\lfloor \frac{8n}{15}\right\rfloor, & {\rm if}\ n\equiv 0,1,5,7,8,9,11,12,13,14~({\rm mod}~ 15)\vspace{1ex}, \\
 \left\lfloor \frac{8n}{15}\right\rfloor-1, & {\rm otherwise.}
\end{array}
\right.$$

(i) Assume that $n\equiv 11~({\rm mod}~15)$, i.e., $n=15m+11$ for some integer $m$.
Applying Lemma \ref{lemma-k-3} to the binary linear $[11,4,5]$ code with one-dimensional hull (see \cite[Table 1]{Kim_preprint}), we have
$$d_{one}(15m+11,4)\geq 8m+5=\left\lfloor \frac{8(15m+11)}{15}\right\rfloor.$$
Combining with the Griesmer bound, we have $d_{one}(n,4)=\left\lfloor \frac{8n}{15}\right\rfloor$ for $n\equiv 11~({\rm mod}~15)$.
A similar argument works for $n\equiv 13~({\rm mod}~15)$.

(ii) Assume that $n\equiv 10~({\rm mod}~15)$, i.e., $n=15m+10$ for some integer $m$.
Applying Lemma \ref{lemma-k-3} to the binary linear $[10,4,4]$ code with one-dimensional hull (see \cite[Table 1]{Kim_preprint}), we have
$$d_{one}(15m+10,4)\geq 8m+4=\left\lfloor \frac{8(15m+10)}{15}\right\rfloor-1.$$
Combining with the Griesmer bound, we have $d_{one}(n,4)=\left\lfloor \frac{8n}{15}\right\rfloor-1$ for $n\equiv 10~({\rm mod}~15)$.
A similar argument works for $n\equiv 2,3,4,6~({\rm mod}~15)$.

(iii) Assume that $n\equiv 7~({\rm mod}~ 15)$, i.e., $n=15m+7$ for some integer $m$. Applying Lemma \ref{lemma-k-3} to the binary linear $[7,4,2]$ code with one-dimensional hull (see \cite[Table 1]{Kim_preprint}), we have
$$d_{one}(15m+7,4)\geq 8m+2=\left\lfloor \frac{8(15m+7)}{15}\right\rfloor-1.$$
Combining with Proposition \ref{prop-7m+6}, we have $d_{one}(n,4)=\left\lfloor \frac{8n}{15}\right\rfloor-1$ for $n\equiv 7~({\rm mod}~15)$.
 A similar argument works for $n\equiv 0,1,5,8,9,12,14~({\rm mod}~ 15)$.
\end{proof}

Next, we consider the exact value of $d_{one}(n,k)$ for $k=n-4$.

\begin{theorem}
Let $n\geq 5$ be an integer. Then
$$d_{one}(n,n-4)=\left\{\begin{array}{ll}
                          4, & {\rm if}~n=5, \\
                          3, & {\rm if}~6\leq n\leq 12, \\
                          2, & {\rm if}~n\geq 13.
                        \end{array}
 \right.$$
\end{theorem}

\begin{proof}
According to Table 1 and \cite[Table 1]{Kim_preprint}, $d_{one}(5,1)=4,$ $d_{one}(n,n-4)=3$ for $6\leq n\leq 12$ and $d_{one}(n,n-4)=2$ for $13\leq n\leq 15$.
It follows from Lemma \ref{lemma-d(n-k)} that $d_{one}(n,n-4)=2$ for $n\geq 16$. This completes the proof.
\end{proof}

\subsection{Optimal binary linear $[n,5]$ and $[n,n-5]$ codes with one-dimensional hull}

First, we recall some known results on $d_{LCD}(n,5)$, which can be found in \cite{AH-BLCD-17-24,ter-11-19}.

\begin{center}
\begin{tabular}{cc|cc|cc}
\multicolumn{6}{c}{{\rm Table 2: Some known results on $d_{LCD}(n,5)$}}\\
\hline
   $n$ & $d_{LCD}(n,5)$& $n$ & $d_{LCD}(n,5)$&$n$ & $d_{LCD}(n,5)$\\
    \hline\hline
    $31m+1$ &$16m-1$ & $31m+13$&$16m+5$ &$31m+24$&$16m+11$ \\

    $31m+5$ &$16m+1$& $31m+17$&$16m+7$&$31m+25$&$16m+11$ \\

    $31m+6$ &$16m+1$& $31m+20$&$16m+9$&$31m+28$&$16m+13$ \\

    $31m+9$ &$16m+3$& $31m+21$&$16m+9$&$31m+29$&$16m+13$ \\
    \hline
\end{tabular}
\end{center}

Combining Corollary \ref{cor-3.9} and Table 2, we have the following table.

\begin{center}
\begin{tabular}{cc|cc|cc}
\multicolumn{6}{c}{{\rm Table 3: Some results on $d_{one}(n,5)$}}\\
\hline
   $n$ & $d_{one}(n,5)$& $n$ & $d_{one}(n,5)$&$n$ & $d_{one}(n,5)$\\
    \hline\hline
   $31m+2$ &$16m$& $31m+14$&$16m+6$&$31m+25$&$16m+12$ \\
   $31m+6$ &$16m+2$& $31m+18$&$16m+8$&$31m+26$&$16m+12$ \\
   $31m+7$ &$16m+2$& $31m+21$&$16m+10$&$31m+29$&$16m+14$ \\
 $31m+10$ &$16m+4$& $31m+22$&$16m+10$&$31m+30$&$16m+14$ \\
    \hline
\end{tabular}
\end{center}

\setlength{\arraycolsep}{1pt}

Assume that $n\geq 7$. By the Griesmer bound, we have
$$d_{one}(n,5)\leq \left\{
\begin{array}{ll}
  \left\lfloor \frac{16n}{31}\right\rfloor, & {\rm if}\ n\equiv 0,1,9,13,15,16,17,21,23,24,25,27,28,29,30~({\rm mod}~ 31)\vspace{1ex}, \\
 \left\lfloor \frac{16n}{31}\right\rfloor-1, & {\rm if}\ n\equiv 2,3,5,6,7,8,10,11,12,14,18,19,20,22,26~({\rm mod}~ 31)\vspace{1ex}, \\
 \left\lfloor \frac{16n}{31}\right\rfloor-2, & {\rm if}\ n\equiv 4~({\rm mod}~ 31).
\end{array}
\right.$$

\begin{prop}\label{prop-31m+}
There is no binary linear $\left[n,5,\left\lfloor \frac{16n}{31}\right\rfloor\right]$ code with one-dimensional hull for $n\equiv 9,24,17,28~({\rm mod}~ 31)$ and $n\geq 7$.
\end{prop}

\begin{proof}
Assume that $n\equiv 9~({\rm mod}~ 31)$, i.e., $n=31m+9$ for some integer $m$.
If $C$ is a binary linear $[31m+9,5,16m+4]$ code with one-dimensional hull for $m\geq 0$, then $C$ is a Griesmer code. By \cite{LCD-lp}, $C$ is self-orthogonal since the minimum distance of $C$ is divisible by $4$. Hence there is no binary linear $\left[n,5,\left\lfloor \frac{16n}{31}\right\rfloor\right]$ code with one-dimensional hull for $n\equiv 9~({\rm mod}~ 31)$. A similar argument works for $n\equiv 17~({\rm mod}~31)$.

Assume that $n\equiv 24~({\rm mod}~ 31)$, i.e., $n=31m+24$ for some integer $m$.
If $C$ is a binary linear $[31m+24,5,16m+12]$ code with one-dimensional hull for $m\geq 1$, then $d(C^\perp)\geq 2$ and there is a vector ${\bf m}=(m_1,\ldots,m_{31})$ such that $C$ is equivalent to $C_5({\bf m})$.
By Lemma \ref{lemma-mi}, we have $m_i\geq m$. Let ${\bf m'}=(m_1-m,\ldots,m_{31}-m)$.
Combining Proposition \ref{prop-rank} and Lemma \ref{lemma-k-3}, the code $C_5({\bf m'})$ is a binary linear $[24,5,12]$ code with one-dimensional hull, which contradicts that $d_{one}(24,5)=11$ (see Table 1).
Hence there is no binary linear $\left[n,5,\left\lfloor \frac{16n}{31}\right\rfloor\right]$ code with one-dimensional hull for $n\equiv 24~({\rm mod}~ 31)$. A similar argument works for $n\equiv 28~({\rm mod}~31)$.
This completes the proof.
\end{proof}

\begin{theorem}
Let $n\geq 7$ be an integer. Then we have
$$d_{one}(n,5)= \left\{
\begin{array}{ll}
  \left\lfloor \frac{16n}{31}\right\rfloor, & {\rm if}\ n\equiv 21,25,29~({\rm mod}~ 31)\vspace{1ex}, \\
 \left\lfloor \frac{16n}{31}\right\rfloor-1, & {\rm if}\ n\equiv 2,3,5,6,7,9, 10,11,14,17, 18,19,20,22,24,26,28,30~({\rm mod}~ 31)\vspace{1ex}, \\
 \left\lfloor \frac{16n}{31}\right\rfloor-2, & {\rm if}\ n\equiv 4~({\rm mod}~ 31).
\end{array}
\right.$$
\end{theorem}

\begin{proof}
(i) Assume that $n\equiv 3~({\rm mod}~31)$, i.e., $n=31m+3$ for some integer $m$. Then we have
$$d_{one}(31m+3,5)\geq d_{one}(31m+2,5)= 16m=\left\lfloor \frac{16(31m+3)}{31}\right\rfloor-1.$$
Combining with the Griesmer bound, we have $d_{one}(n,5)=\left\lfloor \frac{16n}{31}\right\rfloor-1$ for $n\equiv 3~({\rm mod}~31)$.

(ii) Assume that $n\equiv 4~({\rm mod}~31)$, i.e., $n=31m+4$ for some integer $m$.
By Table 3, there is a binary linear $[37,5,18]$ code $C$ with one-dimensional hull.
Let $G$ be the generator matrix of $C$. Since $37>32$, $G$ has the same two columns.
By Proposition \ref{prop-4.9}, there is a binary linear $[35,5,16]$ code with one-dimensional hull.
Applying Lemma \ref{lemma-k-3} to the binary linear $[35,5,16]$ code with one-dimensional hull, we have
$$d_{one}(31m+4,5)\geq 16m=\left\lfloor \frac{16(31m+4)}{31}\right\rfloor-2.$$
Combining with the Griesmer bound, we have $d_{one}(n,5)=\left\lfloor \frac{16n}{31}\right\rfloor-2$ for $n\equiv 4~({\rm mod}~31)$.

(iii) Assume that $n\equiv 5~({\rm mod}~31)$, i.e., $n=31m+5$ for some integer $m$.
Applying Lemma \ref{lemma-k-3} to the binary linear $[36,5,17]$ code with one-dimensional hull (see BKLC \cite{codetables}), we have
$$d_{one}(31m+5,5)\geq 16m+1=\left\lfloor \frac{16n}{31}\right\rfloor-1.$$
Combining with the Griesmer bound, we have $d_{one}(n,5)=\left\lfloor \frac{16n}{31}\right\rfloor-1$ for $n\equiv 5~({\rm mod}~31)$.

(iv) Assume that $n\equiv 9~({\rm mod}~31)$, i.e., $n=31m+9$ for some integer $m$.
Applying Lemma \ref{lemma-k-3} to the binary linear $[9,5,3]$ code with one-dimensional hull (see \cite[Table 1]{Kim_preprint}), we have
$$d_{one}(31m+9,5)\geq 16m+3=\left\lfloor \frac{16n}{31}\right\rfloor-1.$$
Combining with Proposition \ref{prop-31m+}, we have $d_{one}(n,5)=\left\lfloor \frac{16n}{31}\right\rfloor-1$ for $n\equiv 9~({\rm mod}~31)$.

(v) Assume that $n\equiv 11~({\rm mod}~31)$, i.e., $n=31m+11$ for some integer $m$.
Applying Lemma \ref{lemma-k-3} to the binary linear $[11,5,4]$ code with one-dimensional hull (see \cite[Table 1]{Kim_preprint}), we have
$$d_{one}(31m+11,5)\geq 16m+4=\left\lfloor \frac{16n}{31}\right\rfloor-1.$$
Combining with the Griesmer bound, we have $d_{one}(n,5)=\left\lfloor \frac{16n}{31}\right\rfloor-1$ for $n\equiv 11~({\rm mod}~31)$.

(vi) Assume that $n\equiv 17~({\rm mod}~31)$, i.e., $n=31m+17$ for some integer $m$.
Applying Lemma \ref{lemma-k-3} to the binary linear $[17,5,7]$ code with one-dimensional hull (see Table 1), we have
$$d_{one}(31m+17,5)\geq 16m+7=\left\lfloor \frac{16n}{31}\right\rfloor-1.$$
Combining with Proposition \ref{prop-31m+}, we have $d_{one}(n,5)=\left\lfloor \frac{16n}{31}\right\rfloor-1$ for $n\equiv 17~({\rm mod}~31)$.

(vii) Assume that $n\equiv 19~({\rm mod}~31)$, i.e., $n=31m+19$ for some integer $m$. Then we have
$$d_{one}(31m+19,5)\geq d_{one}(31m+18,5)= 16m+8=\left\lfloor \frac{16(31m+19)}{31}\right\rfloor-1.$$
Combining with the Griesmer bound, we have $d_{one}(n,5)=\left\lfloor \frac{16n}{31}\right\rfloor-1$ for $n\equiv 19~({\rm mod}~31)$.

(viii) Assume that $n\equiv 20~({\rm mod}~31)$, i.e., $n=31m+20$ for some integer $m$. Applying Lemma \ref{lemma-k-3} to the binary linear $[20,5,9]$ code with one-dimensional hull (see Table 1), we have
$$d_{one}(31m+20,5)\geq 16m+9=\left\lfloor \frac{16n}{31}\right\rfloor-1.$$
Combining with the Griesmer bound, we have $d_{one}(n,5)=\left\lfloor \frac{16n}{31}\right\rfloor-1$ for $n\equiv 20~({\rm mod}~31)$.

(ix) Assume that $n\equiv 24~({\rm mod}~31)$, i.e., $n=31m+24$ for some integer $m$. Applying Lemma \ref{lemma-k-3} to the binary linear $[24,5,11]$ code with one-dimensional hull (see Table 1), we have
$$d_{one}(31m+24,5)\geq 16m+11=\left\lfloor \frac{16n}{31}\right\rfloor-1.$$
Combining with Proposition \ref{prop-31m+}, we have $d_{one}(n,5)=\left\lfloor \frac{16n}{31}\right\rfloor-1$ for $n\equiv 24~({\rm mod}~31)$. A similar argument works for $n\equiv 28~({\rm mod}~31)$.

Combining (i)-(ix) and Table 3, we obtain the desired result.
\end{proof}

\setlength{\arraycolsep}{4pt}

\begin{theorem}
Let $n\geq 7$ be an integer. Then we have
$$d_{one}(n,5)\geq \left\{
\begin{array}{ll}
 \left\lfloor \frac{16n}{31}\right\rfloor-1, & {\rm if}\ n\equiv 13,15,23,27~({\rm mod}~ 31)\vspace{1ex}, \\
 \left\lfloor \frac{16n}{31}\right\rfloor-2, & {\rm if}\ n\equiv 0,1,8,12,16~({\rm mod}~ 31).
\end{array}
\right.$$
\end{theorem}

\begin{proof}
By Table 1 and \cite[Table 1]{Kim_preprint}, there is a binary linear $[n,5,d]$ code with one-dimensional hull for $(n,d)\in S=\{(8,2),(12,4),(13,5),(15,6),(16,6),(23,10),(27,12),(31,\\14),(32,14)\}$.
For $(n_0,d_0)\in S$, applying Lemma \ref{lemma-k-3} to the binary linear $[n_0,5,d_0]$ code with one-dimensional hull, we obtain the desired result.
\end{proof}

Next, we consider the exact value of $d_{one}(n,k)$ for $k=n-5$.

\begin{theorem}
Let $n\geq 6$ be an integer. Then
$$d_{one}(n,n-5)=\left\{\begin{array}{ll}
                          6, & {\rm if}~n=6, \\
                          4, & {\rm if}~n\in \{7,8,10,12\}, \\
                          3, & {\rm if}~n\in \{9,11,13,14,\ldots,27\},\\
                          2, & {\rm if}~n\geq 28.
                        \end{array}
 \right.$$
\end{theorem}

\begin{proof}
By Table 1 and \cite[Table 1]{Kim_preprint}, $d_{one}(6,1)=6,$ $d_{one}(n,n-5)=4$ for $n\in \{7,8,10,12\}$, $d_{one}(n,n-5)=3$ for $n\in \{9,11,13,14,\ldots,27\}$ and $d_{one}(n,n-5)=2$ for $n\in \{28,29,30\}$.
By Lemma \ref{lemma-1}, $d_{one}(31,26)\leq d_{LCD}(30,25)=2$ (see \cite{ter-11-19}).
By the proof of Lemma \ref{lemma-d(n-k)}, there is a binary linear $[31,26,2]$ code with one-dimensional hull. Thus $d_{one}(31,26)=2$.
It follows from Lemma \ref{lemma-d(n-k)} that $d_{one}(n,n-5)=2$ for $n\geq 32$. This completes the proof.
\end{proof}

\section{Conclusion}

We have studied some properties of binary linear codes with one-dimensional hull, and have established the connection between them and binary LCD codes. We have completely determined the values of $d_{one}(n,k)$ and $d_{one}(n,n-k)$ for $k\leq 5$ except for some special types. For there special types, good lower bounds on $d_{one}(n,5)$ are given. Furthermore, we have extended Kim's result~\cite{Kim_preprint} on $d_{one}(n,k)$ ($1 \le k \le n \le 13)$ to lengths up to 30.\\

\noindent{\bf Acknowledgement:} The research of Shitao Li and Minjia Shi is supported by the National Natural Science Foundation of China (12071001). The research of Jon-Lark Kim is supported by the National Research Foundation of Korea (NRF) Grant funded by the Korea government (NRF-2019R1A2C1088676).\vspace{0.3cm}\\
\noindent{\bf Conflict of Interest:}
The authors have no conflicts of interest to declare that are relevant to the content of this paper.\vspace{0.3cm}\\
\noindent{\bf Data Deposition Information:} The data that support the findings of this study are available at {\tt https://cicagolab.sogang.ac.kr/cicagolab/2660.html}.

\end{document}